\documentclass[10pt,a4]{scrartcl}
\usepackage[T1]{fontenc}
\usepackage{amsfonts}
\usepackage{amsmath}
\usepackage{amssymb}
\usepackage{stmaryrd}
\usepackage{multicol}
\usepackage[inline]{enumitem}
\usepackage{relsize}
\usepackage[thmmarks]{ntheorem}
\theoremnumbering{arabic}
\theoremstyle{plain}
\theorembodyfont{\normalfont}
\newtheorem{theorem}{Theorem}

\theoremnumbering{arabic}
\theoremstyle{plain}
\theorembodyfont{\normalfont}
\newtheorem{lemma}{Lemma}

\newtheorem{remark}{Remark}

\newtheorem{fact}{Fact}

\newtheorem{corollary}{Corollary}

\theoremnumbering{arabic}
\theoremstyle{plain}
\theorembodyfont{\normalfont}
\newtheorem{definition}{Definition}

\theoremnumbering{arabic}
\theoremstyle{plain}
\theorembodyfont{\normalfont}

\theoremstyle{nonumberplain}
\theoremheaderfont{\itshape}
\theorembodyfont{\normalfont}
\theoremsymbol{\ensuremath{\square}}
\newtheorem{proof}{Proof.}

\usepackage{listings}

\usepackage{hyperref}
\newcommand{\myref}[1]{\autoref{#1}}

\usepackage[style=numeric,backend=biber,
maxcitenames = 2,
maxbibnames = 5,
isbn=false,
doi=false,
uniquename=false,
mincrossrefs=100,
uniquelist=false,
backref=false,
maxcitenames=1]{biblatex}
\AtBeginBibliography{\small}

\newcommand{\ndef}[1]{\textbf{#1}}
\newcommand{\nrule}[1]{\ensuremath{\textup{\textsc{#1}}}}

\newcommand{\M}[1]{\ensuremath{\mathsf{#1}}}
\newcommand{\denot}[1]{\ensuremath{\llbracket #1\rrbracket}}

\newcommand{\dom}{\ensuremath{\textup{\textit{dom\,}}}}
\newcommand{\update}[2]{\ensuremath{[#1\mapsto{}#2]}}

\newcommand{\impl}{\ensuremath{\rightarrow}}
\renewcommand{\iff}{\ensuremath{\leftrightarrow}}
\newcommand{\monto}{\ensuremath{\xrightarrow{mon}}}
\newcommand{\Prop}{\ensuremath{\textsf{\textbf{P}}}}

\newcommand{\slist}[1]{\ensuremath{\overline{#1}}}
\newcommand{\incl}{\ensuremath{\subseteq}}
\newcommand{\set}[1]{\ensuremath{\{#1\}}}
\newcommand{\cc}{;}

\newcommand{\Exp}{\ensuremath{\M{Exp}}}
\newcommand{\Term}{\ensuremath{\M{Term}}}
\newcommand{\Fun}{\ensuremath{\mathcal{F}}}
\newcommand{\Act}{\ensuremath{\mathcal{A}}}

\newcommand{\Var}{\ensuremath{\mathcal{V}}}
\newcommand{\Val}{\ensuremath{\mathbb{V}}}
\newcommand{\Valb}{\ensuremath{\mathbb{V}_{\bot}}}

\newcommand{\Evt}{\ensuremath{\mathcal{E}}}

\newcommand{\fv}{\ensuremath{\textup{fv}}}
\newcommand{\valtobool}{\ensuremath{\beta}}
\newcommand{\Bool}{\ensuremath{\mathbb{B}}}
\newcommand{\BoolT}{\ensuremath{\mathbf{true}}}
\newcommand{\BoolF}{\ensuremath{\mathbf{false}}}
\newcommand{\opeval}[2]{\ensuremath{\denot{#1}\,#2}}

\newcommand{\EVE}{\ensuremath{\varnothing}}

\newcommand{\extevt}{\ensuremath{\alpha}}
\newcommand{\extexp}{\ensuremath{\eta}}
\newcommand{\evt}{\ensuremath{\phi}}
\newcommand{\configs}{\ensuremath{\Sigma}}
\newcommand{\config}{\ensuremath{\sigma}}
\newcommand{\closures}{\ensuremath{\mathcal{C}}}

\newcommand{\ilTx}[1]{\texttt{#1}}
\newcommand{\ilIn}{\ilTx{in}}
\newcommand{\ilGoto}[2]{#1\,#2}
\newcommand{\ilLet}[3]{\ilTx{let}\,#1=#2\,\ilIn\,#3}
\newcommand{\ilIf}[3]{\ilTx{if}\,{#1}\,\ilTx{then}\,{#2}\,\ilTx{else}\,{#3}}
\newcommand{\ilFDefs}[3]{\slist{#1\,#2\,=\,#3}}

\newcommand{\ilLetRecP}[2]{\ilTx{fun}~#1~\ilIn~#2}
\newcommand{\ilLetRec}[4]{\ilLetRecP{{#1}\,{#2}={#3}}{#4}}
\newcommand{\ilLetRecM}[4]{\ilLetRecP{\ilFDefs{#1}{#2}{#3}}{#4}}

\newcommand{\ilEvent}[4]{\ilTx{let}~#1=#2\,#3~\ilIn~#4}

\newcommand{\fstate}[3]{\ensuremath{#1~|~#2~|~#3}}

\newcommand{\fevals}{\ensuremath{\longrightarrow}}

\newcommand{\fevalstgsl}[4]{\ensuremath{#1~~\stackrel{#4}{\fevals}_{#3}~~#2}}
\newcommand{\fevalsg}[2]{\stackrel{#2}{\fevals}_{#1}}

\usepackage{bussproofs}
\usepackage{xparse}
\newsavebox{\topprooftreebox}
\newlength{\topprooftreewidth}

\NewDocumentEnvironment{topprooftree}{m}%
  {\begin{lrbox}{\topprooftreebox}}%
  {
   \DisplayProof
   \end{lrbox}
\setlength{\topprooftreewidth}{\wd\topprooftreebox}%
\begin{minipage}[b][][b]{\topprooftreewidth}%
      \nrule{\small #1}\\[1mm]\usebox{\topprooftreebox}\vspace{0pt}
    \end{minipage}
  }

\DeclareDocumentCommand \mkdef { o m } {%
  \IfNoValueTF {#1} {%
    \reversemarginpar\marginpar{#2}\fbox{#2}
  }{%
    \reversemarginpar\marginpar{#1}\fbox{#2}
  }%
}

\newcommand{\res}{\mathit{res}}
\newcommand{\redsys}[5]{\ensuremath{(#1, #2, #3, #4, #5)}}
\newcommand{\LTS}[3]{\ensuremath{(#1, #2, #3)}}
\newcommand{\red}{\ensuremath{\fevals}}

\newcommand{\SimN}{\ensuremath{\raisebox{-2pt}{\ensuremath{\overset{\mathsmaller{<}}{\sim}}}}}
\newcommand{\Sim}{\ensuremath{\mathrel{\SimN}}}
\newcommand{\Bisim}{\vphantom{\Sim}\ensuremath{\sim}}
\newcommand{\BisimN}{\ensuremath{{\sim}}}
\newcommand{\GSimN}{\ensuremath{\approx}}

\newcommand{\trmg}[2]{\ensuremath{\Downarrow_{#2}{#1}}}

\newcommand{\todo}[1]{}
\renewcommand{\todo}[1]{\marginpar{\color{red} TODO: {#1}}}

\setlist[itemize]{leftmargin=*,itemsep=2pt}

\begin{document}

\setlength{\pdfpageheight}{\paperheight}
\setlength{\pdfpagewidth}{\paperwidth}




\title{An Inductive Proof Method for Simulation-based Compiler Correctness}
\author{Sigurd Schneider, Gert Smolka, Sebastian Hack\\Saarland Informatics Campus, Saarland University}



\maketitle

\begin{abstract}
We study induction on the program structure as a proof method for bisimulation-based compiler correctness.
We consider a first-order language with mutually recursive function definitions, system calls, and
an environment semantics.
The proof method relies on a generalization of compatibility of function definition with the bisimulation.
We use the inductive method to show correctness of a form of dead code elimination.
This is an interesting case study because the transformation removes function, variable, and parameter
definitions from the program.
While such transformations require modification of the simulation in a coinductive proof, the inductive method deals with them naturally.
All our results are formalized in Coq.
\end{abstract}

\section{Introduction}
We study induction on the program structure as a proof method for bisimulation-based compiler correctness.
We detail inductive equivalence proofs for the language IL~\cite{DBLP:conf/itp/SchneiderSH15} with respect to a simple, coinductively defined bisimulation.
IL is a first-order language with lexically scoped variables, system calls, mutual recursion, and an environment semantics in the style of Standard ML~\cite{Milner97}.
The restriction to first-order simplifies the setup of semantics, simulation, and the inductive method.
System calls realize (two-way) communication with a system environment
and warrant a bisimulation-based notion of program equivalence.

We explain the inductive method by providing paper proofs of the crucial lemma,
and a case study that applies the lemma in proofs of (bi)similarities:
First, we show contextuality of a bisimulation-based equivalence.
Second, we prove correctness of dead code elimination (DCE), which we split into the following two transformations:
\begin{enumerate}
 \item Unreachable Code Elimination (UCE)
 \item Dead Variable Elimination (DVE)
\end{enumerate}
UCE removes unreachable function definitions and unreachable conditional branches.
DVE removes dead variable definitions and dead parameters from functions.

We verify correctness with respect to a coinductively defined bisimulation.
Intuitively, two configurations are bisimilar if they perform the same (possibly infinite) trace of system calls.
Two configurations are similar, if every trace of the left-hand configuration is also a (partial) trace of the right-hand configuration.
We show that UCE respects bisimilarity, and that DVE respects simulation.

We use induction on the program structure as a proof method for the correctness arguments.
In the inductive proof, the fact that optimizations UCE and DVE remove program statements is no issue.
This is in contrast to a coinductive proof method.
The coinductive hypothesis cannot be applied without further justification if, for example, the source program reduced, but the target program did not.
Such situations naturally arise, for example, if a variable or function definition is removed.
The standard solution is index the simulation with a well-founded relation, and allow stutter steps if the well-founded relation is decreased~\cite{Leroy-Compcert-CACM}.
The main feature of our inductive method is that it does not require modifications of the simulation.

The inductive proof method is enabled by a lemma that strengthens the inductive hypothesis.
The plain structural inductive hypothesis cannot readily be used to show the function definition case correct.
Suppose $\Bisim$ is a semantically defined bisimilarity relation for programs.
The problem with function definitions is that the following rule of congruence does not directly follow by induction.
Function definition introduces a fixed-point of the semantics of
$s$, which requires extra (coinductive) treatment.

\begin{prooftree}
\AxiomC{$s \Bisim s'$}
\AxiomC{$t' \Bisim t'$}
\BinaryInfC{$\ilLetRec{f}{\slist{x}}{s}{t} \Bisim \ilLetRec{f}{\slist{x}}{s'}{t'}$}
\end{prooftree}
This rule of congruence (which we prove in \myref{sec:ctx}) is a special case of a general lemma we prove, and which we use to strengthen the inductive hypothesis.
Our lemma generalizes to optimizations that change function signatures and remove and rename function and variable definitions.
In the latter case, the bodies of related functions are generally not equivalent, but their function applications are, provided the arguments are in some other, often asymmetric, relation.

This paper is accompanied by a Coq development which contains formal proofs of all lemmas and theorems.
The Coq development is part of a larger compiler verification project and available online:
\begin{center}
\url{www.ps.uni-saarland.de/~sdschn/lvc-ind/}
\end{center}

\subsection{Contributions}
The paper makes the following contributions:
\begin{enumerate}
  \item We develop a simple proof method for IL which supports induction on the program structure for proofs of bisimulation-based equivalence. IL has mutually recursive function definitions and system calls.
  \item We detail the method in the proof of contextuality of our simulation, and in the correctness proofs of Unreachable Code Elimination and Dead Variable Elimination.
  We explain how the method deals with the removal of definitions, function parameters, and mutual recursion.
  \item The correctness proofs of UCE and DVE are carried out formally in Coq in a setting with De-Bruijn function binders and mutual recursion.
    Here in the paper, we present a named version in hope for better readability.
\end{enumerate}

\subsection{Related Work}
\textbf{Inductive Proofs for Bisimulations}
The idea to use compatibility lemmas to simplify correctness proofs is outlined
in \S 2.8 of the master thesis of one of the authors~\cite{Schneider13}.
The master thesis uses a version of IL without mutual recursion and system calls.
A basic version of the extension lemma, which enables the inductive method and which we prove in \myref{sec:ext_fun_ctx}, appears in the master thesis as Lemma 3.
The masters thesis uses the extension lemma to show that contextual equivalence is characterized by a simulation-based definition.
In \myref{sec:ctx} of this paper we show that a bisimulation-based definition is sound for contextual equivalence.

Neis et al.~\cite{DBLP:conf/icfp/NeisHKMDV15} recently used an inductive method to deal with stuttering steps when proving their elaborate parametric inter-language simulations (PILS). In the PILS framework, they verify a compiler for an imperative higher-order language with non-mutually recursive functions that take a fixed number of arguments.

Neis et al.\ use an inductive method to deal with stuttering steps in, among other things, the correctness proof of a form of DCE with respect to PILS that only eliminates unused let-bindings.
Neis et al.\ mention that their framework provides a series of compatiblity lemmas simplifying the proof, but do not state the precise form of the lemmas in the paper.
Our DCE removes dead function parameters, unused function definitions, and unreachable branches of conditionals.

As Neis et al. deal with a higher-order language in the PILS framework, their setup is necessarily more complicated than ours, and they only give a high-level description of how they setup the induction.
This paper aims to explain the inductive method in a simple setting that still allows to see its merit.
We include mutual recursion as it directly interacts with the setup of the inductive method.
We detail the proofs in the hope to expose the inductive method in general, independent of a framework.

\textbf{DCE in CompCert}
Dead code elimination (DCE) in CompCert~\cite{Leroy-backend} is carried out by two optimizations.
First, the translation from RTL to LTL replaces instructions that write to dead registers with no-ops.
Second, the branch tunneling phase removes no-ops.

In CompCert, optimizations that remove instructions are proven correct via a measure argument
that justifies applicability of the coinductive hypothesis.
Dealing with a measure is more complicated than a plain coinductive proof.
For this reason, only the branch tunneling phase removes instructions.
Our inductive approach supports removal of definitions (i.e. instructions) without additional effort in the correctness proof of any optimization and does not require an additional measure.

\textbf{Correctness Arguments in Verified Compilers}
The correctness arguments in VeLLVM~\cite{DBLP:conf/pldi/ZhaoNMZ13}, the verified LLVM project,
CompCertSSA~\cite{DBLP:conf/esop/BartheDP12}, and CompCertTSO~\cite{DBLP:journals/jacm/SevcikVNJS13}
use exclusively coinduction for correctness proofs.
Those compilers operate on a graph-based program representation, so induction on the program structure is not as useful as in our term-based setting.

\textbf{Howe's Method}
Howe's Method~\cite{DBLP:journals/iandc/Howe96} is a general method to show that a (coinductively defined) relation is a congruence.
Howe's method is particularly effective in higer-order settings.
Howe's method first constructs a precongruence candidate relation that contains bisimilarity and can easily be shown to be a congruence.
Afterwards the candidate relation is shown to coincide with bisimilarity.
In our work we prove by coinduction that for showing two function definitions bisimilar
 it suffices to show that their bodies are bisimilar, assuming that all related functions in the environment are bisimilar.
Howe's method seems to be geared towards congruence properties and we are not aware of work extending it to optimizations that change function signatures.

\textbf{CakeML}
CakeML~\cite{DBLP:conf/esop/OwensMKT16} is a verified compiler for a substantial subset of Standard ML.
CakeML originally uses big-step semantics, which does not account for diverging behaviors.
Recently, CakeML switched to an evaluation function with a step limit to specify the semantics.
Both approaches directly support inductive proofs on the semantics.

\subsection{Outline}
The paper is organized as follows.
We define the syntax and semantics of the language IL in \myref{sec:il_sem}.
In \myref{sec:progeq} we repeat the definition of program equivalence from previous work~\cite{DBLP:conf/itp/SchneiderSH15},
and give a new characterization using parameterized co-induction~\cite{DBLP:conf/popl/HurNDV13}.
We then prove compatibility rules admissible that we use repeatedly in the following proofs.
We develop the inductive method in \myref{sec:ind_method} and use it in \myref{sec:ctx} to show
that the bisimulation we defined is contextual.
We describe how program analysis information is represented in our framework in \myref{sec:ann}.
We specify reachability and prove unreachable code elimination correct in \myref{sec:UCE}.
We specify true liveness and prove dead variable elimination correct in \myref{sec:DVE}.
We discuss the formal development in \myref{sec:coq} and conclude in \myref{sec:conclusion}.


\section{IL}
\label{sec:il_sem}
\subsection{Values, Variables, and Expressions}
We assume a type $\Val$ of values and a function
$\valtobool:\Val\to\Bool=\set{\BoolT,\BoolF}$ that we use to simplify the semantic rule for the conditional.
By convention, $v$ ranges over $\Val$.
We use the countably-infinite alphabet~$\Var$ for names $x,y,z$ of values, which we call \emph{variables}.

\newcommand{\envle}{\ensuremath{\sqsubseteq}}
We assume a type $\Exp$ of expressions.
By convention, $e$~ranges over \Exp.
Expressions are pure, their evaluation is deterministic and may fail, hence
expression evaluation is a function $\opeval{\cdot}{}:\Exp\to(\Var\to\Valb)\to\Valb$.
Environments are of type $\Var\to\Valb$ to track uninitialized variables, and are
partially ordered by $\envle$, which is the pointwise lifting of the relation defined by the two equations $\bot \envle w$ and $w \envle w$, where $w\in\Valb$.
We assume that expression evaluation is monotone\label{inl:exp_mon}, i.e., $V \envle V' \impl \opeval{e}\,V \envle \opeval{e}\,V'$.

We assume a function $\fv:\Exp\to\M{set}\,\Var$ such that for all environments $V,V'$
that agree on $\fv(e)$ we have $\opeval{e}{V}=\opeval{e}{V'}$.
We use the notation $\slist{x}$ for a list of variables.
We lift $\opeval{\cdot}{}$ pointwise to lists of expressions in a strict fashion: $\opeval{\slist{e}}$ yields a list of values if none of the expressions in~$\slist{e}$ failed to evaluate, and $\bot$ otherwise.

We sometimes omit the side condition $\opeval{\slist{e}}{V}\not=\bot$ in the presentation if
$\opeval{\slist{e}}{V}$ is used in a place where type $\Val$ is required.
For example, we write $\beta(\opeval{\slist{e}}{V})=\BoolT$ instead of
$\exists v:\Val, \opeval{\slist{e}}{V}=v \land \beta v = \BoolT$.

\label{chap:il}
\subsection{Syntax}
IL is a first-order language with a tail-call restriction, mutual recursion, and system calls.
IL syntactically enforces a first-order discipline by using a separate alphabet $\Fun$ for function names $f,g,h$.
Variables are lexically scoped binders, and a function definition creates a closure that captures variables.

IL uses a third alphabet $\Act$ for names $\alpha$ which we call \emph{actions}.
The term $\ilLet{x}{\extevt(\slist{e})}{\dots}$ is like a system call $\alpha$ with argument list $\slist{e}$ that non-deterministically returns a value.

IL allows mutually recursive function definitions.
The syntax of IL is given in \myref{fig:ilf_syntax}.

\newcommand{\mkBlocks}[1]{\ensuremath{\llparenthesis#1\rrparenthesis}}
\newcommand{\cg}[1]{\ensuremath{[#1]}}
\newcommand{\mkGroup}[1]{\ensuremath{\cg{#1}}}

\begin{figure}
  \centering
\begin{align*}
\eta ::=&~e~|~\extevt(\slist{e})&&\quad\textup{extended expression}\\
\Term\ni{}s,t ::=&~\ilLet{x}{\extexp}{s}&&\quad\textup{variable binding}\\
   |&~\ilIf{e}{s}{t}&&\quad\textup{conditional}\\
   |&~e&&\quad\textup{expression}\\
   |&~\ilLetRecM{f}{\slist{x}}{s}{t}&&\quad\textup{function definition}\\
   |&~f\,\slist{e}&&\quad\textup{application}
\end{align*}
  \caption{Syntax of IL}
  \label{fig:ilf_syntax}
\end{figure}

\newcommand{\bnfeq}{\ensuremath{\mathrel{:\mkern-2mu:\!\!\shorteq}}}
\newcommand{\defeq}{\ensuremath{\mathrel{:\!\!\shorteq}}}
\newcommand{\LE}{\ensuremath{L}}
\newcommand{\rewind}[2]{\ensuremath{#1^{-#2}}}

\begin{figure}
\begin{center}
\begin{topprooftree}{Op}
   \AxiomC{$\denot{e}\,V=v$}
   \UnaryInfC{\fevalstgsl{\fstate{\LE}{V}{\ilLet{x}{e}{s}}}
                      {\fstate{\LE}{V\update{x}{v}}{s}}{}{}
                      }
\end{topprooftree}
\end{center}
\begin{center}
\begin{topprooftree}{Cond}
   \AxiomC{$\denot{e}\,V=v$}
   \AxiomC{$\valtobool(v) = b$}
   \LeftLabel{{\small\nrule{}}}
   \BinaryInfC{\fevalstgsl{\fstate{\LE}{V}{\ilIf{e}{s_\BoolT}{s_\BoolF}}}
                       {\fstate{\LE}{V}{s_i}}{}{}}
\end{topprooftree}
\end{center}
\begin{center}
\begin{topprooftree}{Extern}
   \AxiomC{$v'\in\Val$}
   \AxiomC{$\denot{\slist{e}}\,V=\slist{v}$}
   \BinaryInfC{\fevalstgsl{\fstate{\LE}{V}{\ilEvent{x}{\alpha(\slist{e})}{\!}{s}}}
                        {\fstate{\LE}{V\update{x}{v'}}{s}}
                        {}{v'=\alpha(\slist{v})}}
\end{topprooftree}
\end{center}
\begin{center}
\begin{topprooftree}{Fun}
  \AxiomC{}
  \UnaryInfC{\fevalstgsl{\fstate{\LE}{V}{\ilLetRecM{f}{\slist{x}}{s}{t}}}
                     {\fstate{\mkBlocks{\ilFDefs{f}{\slist{x}}{s}}_V\cc \LE}{V}{t}}{}{}}
\end{topprooftree}
\end{center}
\begin{center}
\begin{topprooftree}{App}
  \AxiomC{$\denot{\slist{e}}\,V=\slist{v}$}
  \AxiomC{$\LE_f = (V', \slist{x}, s)$}
  \BinaryInfC
    {\fevalstgsl{\fstate{\LE}{V}{\ilGoto{f}{\slist{e}}}}
             {\fstate{\rewind{\LE}{f}}{{V'}\update{\slist{x}}{\slist{v}}}{s}}{}{}}
\end{topprooftree}
\end{center}
\begin{align*}
\mkBlocks{\ilFDefs{f}{\slist{x}}{s}}_V =\mkGroup{\slist{f:(V,\slist{x},s)}}
\end{align*}
\caption{Semantics of IL}
\label{fig:ili-sem}
\end{figure}

\subsection{Semantics}
\label{sec:ilf}
The semantics of IL is given as small-step relation $\fevals$ in \myref{fig:ili-sem}.
Note that the tail-call restriction ensures that no call stack is required.
The reduction relation $\fevals$ operates on \ndef{configurations} of the form $(L,V,s)$ where $s$ is the IL term to be evaluated.
The semantics does not rely on substitution, but uses an environment $V:\Var\to\Valb$ for variable definitions and a context $L$ of function definitions.
Transitions in~$\fevals$ are labeled with \emph{events} $\evt$.
By convention, $\psi$ ranges over events different from~$\tau$.
 $$\Evt\ni\evt ::= \tau~~|~~v=\alpha(\slist{v})$$
The silent event is denoted by $\tau$, and we omit it by convention.

A \ndef{context} is a list of groups of named definitions.
For example, the context $K=\cg{f_1:a_1, f_2:a_2}\cc\cg{g_1:b_1}$ consists of three definitions in two groups.
We define a function $\dom$ that yields the domain of a context as a list, e.g. $\dom{K}=f_1,f_2,g_1$.
A definition in a context may refer to previous definitions and definitions in its group.
Notationally, we use contexts like functions:
To access the first element with name $f$, we write \mkdef{$L_f$} and
we have $L_f=\bot$ if no such element exists.
We write $\rewind{L}{f}$ for the context obtained from $L$ by dropping all groups before the first group containing $f$.
We write $\cc$ for context concatenation
and $\EVE$ for the empty context.

A \ndef{closure} is a tuple $(V,\slist{x},s)\in\closures$ consisting of an environment $V$, a parameter list $\slist{x}$, and a function body $s$.
Since a function~$f$ in a context can only refer to previously defined functions and functions in its own group, the first-order restriction allows the closures to be non-recursive: function closures do not need to close under functions.
An application $f \slist{e}$ causes the function context $L$ to rewind to $\rewind{L}{f}$, i.e. up to the group with the definition of $f$ (rule \nrule{App}).
In contrast to higher-order formulations, we do not define closures mutually recursively with the values of the language.

A \textbf{system call} $\ilEvent{x}{\alpha}{\slist{e}}{s}$ invokes a function~$\alpha$ of the system, which is not assumed to be deterministic.
This reflects in the rule \nrule{Extern}, which does not restrict the result value of the system call other than requiring that it is a value.
The transition records the system call name $\alpha$, the argument values $\slist{v}$ and the result value $v'$ in the event $v'=\alpha(\slist{v})$.

\section{Program Equivalence}
\label{sec:progeq}
Before any transformation can be proven correct, we must formally
define what semantic equivalence means.
Semantic equivalence is not directly tight to the language, but only
to the way the language interacts with its environment.
In our case, the language interacts with the environment via system calls,
and possibly a result value.
We abstract this behavior with internally deterministic reduction systems (IDRS),
that we previously introduced \cite{DBLP:conf/itp/SchneiderSH15}.
\begin{definition}
A \emph{reduction system} (RS) is given by a tuple
$\redsys{\configs}{\Evt}{\red}{\tau}{\res}$ such that
\begin{multicols}{2}
\begin{enumerate}
\item $\LTS{\configs}{\Evt}{\red}$ is a LTS
\item $\res:\configs\to\Valb$
\item $\res\,\config=v \Rightarrow \config~\textup{$\,\not\!\!\red$}$
\item $\tau\in\Evt$
\end{enumerate}
\end{multicols}
\noindent
An \emph{internally deterministic} reduction system (IDRS) additionally satisfies
\begin{enumerate}
\setcounter{enumi}{4}
\item $\config\fevalsg{}{\phi}\config_1 \land \config\fevalsg{}{\phi}\config_2 \Rightarrow \config_1 = \config_2$ \hfill \textup{action-deterministic} \quad

\item $\config\fevalsg{}{\phi}\config_1 \land \config\fevalsg{}{\tau}\config_2 \Rightarrow \phi = \tau$ \hfill \textup{$\tau$-deterministic} \quad
\end{enumerate}

\end{definition}

The semantics of IL forms an IDRS: We define $\res$ such that $\res(\config)=v$ if $\config$ is of the form $(F,V,e)$ and $\denot{e}\,V=v$. Otherwise, $\res(\config)=\bot$.

\subsection{Similarity and Bisimilarity}
\label{sec:sim_bisim_std}
\label{sec:bisimilarity}
To define what it means that two IDRS behave equivalently, we use (bi)similarity.
Bisimilarity is obtained as the greatest fixed-point, and naturally accounts for
diverging behaviors.
In previous work \cite{DBLP:conf/itp/SchneiderSH15} we have given a definition of the bisimilarity relation we present here,
and we have shown that it sound and complete for trace equivalence.

Before we give the rules defining (bi)similarity, need some definitions.
We write $\config\trmg{w}{}$ (where $w\in\Valb$) if $\config$ terminates with $w$, that is,
 $\config\red^\ast\config'$ such that $\config'$ is $\red$-terminal and $\res(\config')=w$.
We also want to be able do identify configurations which are about to execute a system call,
and say that such configurations are \emph{ready}.
Finally we introduce the notation
$\config_1 \stackrel{R}{\leadsto}\config_2$ for the standard forward-simulation
property.
That is, every transition $\config_1$ takes can also be taken
by $\config_2$, and the two successor configurations are related by $R$.
%
\begin{figure}
\begin{center}
\begin{topprooftree}{Bisim-Silent}
  \AxiomC{$\config_1\fevalsg{}{}^{+}\config_1'$}
  \AxiomC{$\config_2\fevalsg{}{}^{+}\config_2'$}
  \AxiomC{$\config_1'\Bisim\config_2'$}
  \doubleLine
  \TrinaryInfC{$\config_1\Bisim\config_2$}
\end{topprooftree}
\begin{topprooftree}{Bisim-Term}
  \AxiomC{$\config_1\trmg{w}{}$}
  \AxiomC{$\config_2\trmg{w}{}$}
  \doubleLine
  \BinaryInfC{$\config_1\Bisim\config_2$}
\end{topprooftree}
\end{center}
\begin{center}
\begin{topprooftree}{Bisim-Extern}
  \Axiom$\fCenter\config_1\fevalsg{}{}^{*}\config_1'$
  \noLine
  \UnaryInf$\fCenter\config_2\fevalsg{}{}^{*}\config_2'$
  \Axiom$\fCenter\config_1', \config_2' ~ \textup{ready}$
  \Axiom$\fCenter\config'_1 \stackrel{\BisimN}{\leadsto}\config'_2$
  \noLine
  \UnaryInf$\fCenter\config'_2 \stackrel{\BisimN}{\leadsto} \config'_1$
  \doubleLine
  \TrinaryInfC{$\config_1\Bisim\config_2$}
\end{topprooftree}
\begin{topprooftree}{\small\nrule{Sim-Error}}
  \AxiomC{$\config_1\fevals^*\config'_1$}
  \noLine
  \AxiomC{$\config'_1\,\textup{terminal}$}
  \UnaryInfC{$\res\,\config'_1=\bot$}
  \doubleLine
  \BinaryInfC{$\config_1\Sim\config_2$}
\end{topprooftree}
\end{center}
\caption{Defining Rules of Similarity and Bisimilarity}
\label{fig:bisim_rules}
\end{figure}
\begin{definition}[Bisimilarity] Let $\redsys{S}{\Evt}{\fevalsg{}{}}{\res}{\tau}{}$ be an IDRS.
We define bisimilarity in type theory as relation of type ${S}\to{S}\to\Prop$, where $\Prop$ is the universe of propositions. Bisimilarity $\BisimN$ is defined coinductively as the greatest relation closed under the rules \nrule{Bisim-Silent}, \nrule{Bisim-Extern}, \nrule{Bisim-Term} in \myref{fig:bisim_rules}.
\label{def:bisim}
\end{definition}
%
\nrule{Bisim-Silent} allows to match finitely many steps on both sides, as long as all transitions are silent.
This makes sense for IDRS, but would not yield a meaningful definition otherwise.
\nrule{Bisim-Extern} ensures that every external transition of $\config_1'$ is matched by the same external transition of $\config_2'$, and vice versa.
This ensures that if two programs are in relation, they react to every possible result value of the external call in a bisimilar way.
$\config_1',\config_2'$ are required to be ready to simplify case distinctions by ensuring that the next event cannot be $\tau$.
\begin{definition}
\label{def:sim}
Let $\redsys{S}{\Evt}{\fevalsg{}{}}{\res}{\tau}{}$ be an IDRS.
Similarity is defined as the greatest relation closed under the rules \nrule{Bisim-Silent}, \nrule{Bisim-Extern}, \nrule{Bisim-Term} and \nrule{Sim-Error} in \myref{fig:bisim_rules}.
\end{definition}
\nrule{Sim-Error} can be used to justify similarity for any configuration on the right side, if the left side can be shown to reduce to a stuck configuration.

\subsection{Bisimilarity as Symmetrization of Similarity}
Obtaining bisimilarity as symmetrization of similarity is useful if
the properties one wants to show are symmetric properties:
Bisimilarity is obtained from a proof of similarity and symmetry.
If the property is not symmetric, one needs two proofs of similarity,
which, in practice, share a lot of arguments.
Leroy \cite{Leroy-backend} and Sevc{\'{\i}}k \cite{DBLP:journals/jacm/SevcikVNJS13}
avoid a second proof for the backward direction
by showing that on the class of LTS they are using, forward and backward simulation coincide.
In our setting, bisimilarity is the basic definition, and simulation is obtained by adding an
``escape'' rule (\nrule{Sim-Error}) that justifies similarity if the left configuration is stuck.
In this way, we can show forward and backward direction in one proof,
but do not require the two directions to be equivalent.
In the presence of non-determinism, the two forward and backward simulation do not coincide.

\section{Parameterized Coinduction}

\newcommand{\SimF}{\ensuremath{\mathit{sim}}}
\newcommand{\BisimF}{\ensuremath{\mathit{bisim}}}
\newcommand{\cofix}{\ensuremath{\textup{cofix}}}

For the formalization, we need to define simulation and bisimulation via parameterized coinduction~\cite{DBLP:conf/popl/HurNDV13} to side-step the too restrictive guardedness check for co-fixed points in Coq.
A coinductively defined function must be productive to be well-formed, a criterion that is dual to the requirement that an inductively defined function must be terminating.
Coq requires co-recursion to occur syntactically directly below a constructor of the co-inductive definition, which is a sufficient criterion for productivity.
Parameterized coinduction allows for productivity to be accounted for in a semantic way.
We recapitulate the basic setup of parameterized coinduction following Hur et al.~\cite{DBLP:conf/popl/HurNDV13} in this section, and outline how parameterized coinduction works in \myref{rem:outline}, when we have all definitions at hand.
\newcommand{\X}{\ensuremath{X}}
\newcommand{\x}{\ensuremath{x}}
\newcommand{\y}{\ensuremath{y}}
\newcommand{\z}{\ensuremath{z}}
\newcommand{\Xtype}{complete prelattice}
\begin{definition}[Complete Prelattice]
A complete prelattice $(\X,\sqsubseteq,\sqcap,\sqcup,\top,\bot)$ is a complete lattice that is defined with respect to $\x\equiv\y:=\x\sqsubseteq\y\land\y\sqsubseteq\x$ instead of equality, i.e. a lattice which does not require anti-symmetry.
\end{definition}
The setup relies on the notion of a complete prelattice.
Hur et al. do not require anti-symmetry, but base the paper presentation on a complete lattice nonetheless.
We apply parameterized coinduction to functions into~$\Prop$, the universe of proprositions.
Function types into $\Prop$ only form a complete lattice, if the axioms of propositional extensionality and functional extensionality are assumed.
The function types into $\Prop$ each form a complete prelattice.

\newcommand{\Lincl}{\ensuremath{\sqsubseteq}}
\newcommand{\Lcup}{\ensuremath{\sqcup}}
\newcommand{\Lcap}{\ensuremath{\sqcap}}

\begin{definition}[Greatest Fixed Point]
\label{def:cofix}
Let $X$ be a complete prelattice.
We define a function
\begin{align*}
 \cofix&~:~(\X\to\X)\to\X\\
 \cofix&~f~:=~\bigsqcup\set{\y\in\X~|~\y \Lincl f \y}
\end{align*}
We use the notations $\nu x.s := \cofix (\lambda x.s)$ and $\nu f := \cofix f$.
\end{definition}

\begin{fact}
\label{fact:cofix_unfold}
Let $\X$ be a \Xtype{} and $f$ be a monotone function.
Then $\cofix f \Lincl f (\cofix f)$.
\end{fact}

\newcommand{\G}{\ensuremath{\textup{\textbf{G}}}}
\newcommand{\Lat}{\ensuremath{C}}

\begin{definition}[Parameterized Greatest Fixed Point]
\label{lem:gf_mon}
\label{def:pgf}
Let $X$ be a \Xtype{} and $f:X \to X$ be a monotone function.
We define a function
\begin{align*}
 \G&~:~(\X \monto \X)\monto \X\monto \X\\
 \G&~f~\x~:=\nu \y.f(\x \Lcup \y)
\end{align*}
It is easy to check that $\G$ and $\G f$ are monotone.
\end{definition}

\begin{lemma}[Initialize]
\label{lem:gf_init}
$\nu f \equiv \G f \bot$.
\end{lemma}

\begin{lemma}[Unfold]
\label{lem:gf_unfold}
$\G f \x \equiv f (\x \Lcup \G f \x)$.
\end{lemma}

\begin{lemma}[Accumulate]
\label{lem:gf_acc}
$\y \Lincl \G f \x \iff \y \Lincl \G f (\x \Lcup \y)$.
\end{lemma}
\begin{proof}
See \cite{DBLP:conf/popl/HurNDV13}.
\end{proof}

\begin{corollary}
\label{lem:gf_ind}
If $\forall z, x \Lincl z \impl y \Lincl z \impl y \Lincl \G f z$ then $y \Lincl \G f x$.
\end{corollary}
\myref{rem:outline} outlines the usage of \myref{lem:gf_ind} as coinductive proof principle.
The definition of $\G$ and its lemmas are provided by the Paco library~\cite{DBLP:conf/popl/HurNDV13}.
The Paco library realizes $\G$ directly as a coinductively defined predicate,
instead of using the cofixed point operator we defined for this presentation in~\myref{def:cofix}.

\section{Similarity and Bisimilarity as Parameterized Greatest Fixed Point}

\newcommand{\SimType}{\ensuremath{\textsf{STy}}}
\newcommand{\STBi}{\ensuremath{\textsf{bisim}}}
\newcommand{\STSi}{\ensuremath{\textsf{sim}}}

We obtain definitions equivalent to similarity and bisimilarity
with the fixed point operator $\G$ from a single function.
The use of a single function allows us to show many properties which hold for both, similarity and bisimilarity, with one lemma.
This saves a lot of repetition particularly in the proof of transitivity.

\begin{definition}
We define the function $\SimF$ that generates similarity and bisimilarity in \myref{fig:sim_functional}.
\end{definition}
\begin{figure}
\begin{align*}
&\SimType\ni{s}::=\STBi~|~\STSi\\
&\SimF:(\SimType\to\Sigma\to\Sigma\to\Prop)\to(\SimType\to\Sigma\to\Sigma\to\Prop)\\
&\SimF\,r\,p\,\config_1\,\config_2:= (\exists w.~ \config_1\trmg{w}{} \land \config_2\trmg{w}{})&&\nrule{Bisim-Term}
\\
  &\,\quad \lor (\exists \config_1'\config_2'.~ \config_1\fevalsg{}{}^{+}\config_1' \land \config_2\fevalsg{}{}^{+}\config_2' \land r\,p\,\config_1'\,\config_2')&&\nrule{Bisim-Step}
\\
  &\,\quad \lor (\exists \config_1'\config_2'.~ \config_1\fevalsg{}{}^{+}\config_1' \land \config_2\fevalsg{}{}^{+}\config_2' \\
  &\,\quad\quad \land \config_1', \config_2' ~ \textup{ready} \land \config'_1 \stackrel{r\,p}{\leadsto}\config'_2 \land \config'_2 \stackrel{r\,p}{\leadsto} \config'_1) &&\nrule{Bisim-Extern}
\\
  &\,\quad \lor (p=\STSi \\
  &\,\quad\quad \land   \exists \config_1'.~ \config_1\fevals^*\config'_1 \land \config'_1\,\textup{terminal} \land \res\,\config'_1=\bot)&&\nrule{Sim-Error}
\end{align*}
\caption{Generating Function for Simulation and Bisimulation. Each disjunct corresponds to a rule from \myref{fig:bisim_rules}.}
\label{fig:sim_functional}
\end{figure}

\begin{remark}[Outline of Parametric Co-Induction]
\label{rem:outline}
A parameterized coinduction using $\SimF$ always has a conclusion of the form $R \incl \G\,\SimF\,r\,p$ for some relations $R$ and $r$.
Applying \myref{lem:gf_ind} sets up the coinduction: We have to show $R \incl \G\,\SimF\,r'\,p$ but can assume $r \incl r'$ and $R \incl r'$.
The assumption $R \incl r'$ is the coinductive hypothesis.
The proof typically proceeds by unfolding $\G$ according to \myref{lem:gf_unfold}: $R \incl \SimF(r'\cup\G\,\SimF\,r')\,p$. Unfolding exposes the generating function $\SimF$, each disjunct of which corresponds to a constructor (cf. \myref{fig:bisim_rules}).
In places where the constructor uses co-recursion, the function $\SimF$ applies its parameter $r$.
In our proof, the parameter is $r'\cup\G\,\SimF\,r'$.
This ensures that the co-hypothesis $R \incl r'$ is only applied after one of the constructors has been ``used''.
The parameter in the definition of~$\G$ encodes the productivity requirement semantically.
\end{remark}

\begin{lemma}
\label{lem:sim_gen_trans}
Let $p:\SimType$ and $\config_1,\config_2,\config_3:\configs$.
If $\G\,{\SimF}\,\bot\,p\,\config_1\,\config_2$
and $\G\,{\SimF}\,\bot\,p\,\config'_2\,\config_3$
and $\config_2\fevalsg{}{}^{*}\config'_2$ or $\config'_2\fevalsg{}{}^{*}\config_2$
then $\G\,{\SimF}\,\bot\,p\,\config_1\,\config_3$.
\end{lemma}
\begin{proof}
The proof is by case analysis on $\G\,{\SimF}\,\bot\,p\,\config_1\,\config_2$
and $\G\,{\SimF}\,\bot\,p\,\config'_2\,\config_3$.
The cases are not difficult, but tedious.
\end{proof}

\newcommand{\GSim}[2]{\ensuremath{\GSimN^{#1}_{#2}}}
\begin{definition}
Let $r:\SimType\to\configs\to\configs\to\Prop$. We define:
\begin{align*}
{\GSim{p}{r}}&:=\G\,{\SimF}\,r\,p\\
\SimN_r&\mathrel{:=}{\GSim{\STSi}{r}}\\
\BisimN_r&\mathrel{:=}{\GSim{\STBi}{r}}
\end{align*}
\end{definition}

\begin{lemma}
$\SimN_\bot$ is a preorder.
\end{lemma}
\begin{proof}
Reflexivity is trivial; transitivity is an instance of \myref{lem:sim_gen_trans}.
\end{proof}

\begin{lemma}
$\BisimN_\bot$ is an equivalence relation.
\end{lemma}
\begin{proof}
Reflexivity and symmetry are trivial; transitivity is an instance of \myref{lem:sim_gen_trans}.
\end{proof}

The following theorem establishes trust in our non-standard setup.
The definitions obtained from parameterized coinduction and the function $\SimF$ are equivalent to the more basic definitions from \myref{sec:sim_bisim_std}.
\begin{lemma}
$\SimN\equiv\SimN_\bot$ and $\BisimN\equiv\BisimN_\bot$.
\end{lemma}

\subsection{Properies of Similarity and Bisimilarity}
The following admissible rules allow us to retract to reduction successors of states when showing (bi)similarity.

\begin{lemma} The rules in \myref{fig:bisimprops} are admissible.
\label{lem:sim_expansion_closed}
\label{lem:sim_Y}
\end{lemma}

\begin{figure}
\begin{center}
\begin{topprooftree}{{\small\nrule{Sim-Expansion-Closed}}}
  \AxiomC{$\config_1 \fevalsg{}{}^* \config'_1$}
  \AxiomC{$\config_2 \fevalsg{}{}^* \config'_2$}
  \AxiomC{$\config'_1 \GSim{p}{r} \config'_2$}
  \TrinaryInfC{$\config_1 \GSim{p}{r} \config_2$}
\end{topprooftree}
\begin{topprooftree}{\small\nrule{Sim-Retract}}
  \AxiomC{$\config_{1} \fevalsg{}{} \config''_1$}
  \noLine
  \UnaryInfC{$\config'_{1} \fevalsg{}{} \config''_1$}
  \AxiomC{$\config_{2} \fevalsg{}{} \config''_2$}
  \noLine
  \UnaryInfC{$\config'_{2} \fevalsg{}{} \config''_2$}
  \AxiomC{$\config_{1} \GSim{p}{r} \config_2$}
  \TrinaryInfC{$\config'_{1} \GSim{p}{r} \config_2'$.}
\end{topprooftree}
\end{center}
\caption{Closedness under Expansion and Retraction}
\label{fig:bisimprops}
\end{figure}

\subsection{Structural Rules}
The following lemmas are formulated with respect to $\GSim{p}{r}$, which allows us to use one proof to show a property of both simulation and bisimulation.


\begin{figure}
  \begin{center}
    \begin{topprooftree}{{\small\nrule{Sim-Let-Op}}}
      \AxiomC{$\denot{e}\,V=\denot{e'}\, V'$}
      \noLine
      \UnaryInfC{$ \forall v, (L,V\update{x}{v},s)~({\GSim{p}{r}}\cup r)~(L', V'\update{x'}{v}, s') $}
      \UnaryInfC{$(L,V,\ilLet{x}{e}{s})\GSim{p}{r}(L', V', \ilLet{x'}{e'}{s'})$}
    \end{topprooftree}
  \end{center}
  \begin{center}
    \begin{topprooftree}{{\small\nrule{Sim-Let-Call}}}
      \AxiomC{$\denot{\slist{e}}\,V=\denot{\slist{e'}}\, V'$}
      \noLine
      \UnaryInfC{$ \forall v, (L,V\update{x}{v},s)~({\GSim{p}{r}}\cup r)~(L', V'\update{x'}{v}, s') $}
      \UnaryInfC{$(L,V,\ilEvent{x}{\!f}{\slist{e}}{s})\GSim{p}{r}(L', V', \ilEvent{x'}{\!f}{\slist{e'}}{s'})$}
    \end{topprooftree}
  \end{center}
  \begin{center}
    \begin{topprooftree}{{\small\nrule{Sim-Cond}}}
      \AxiomC{$\denot{e}\,V=\denot{e'}\, V'$}
      \noLine
      \UnaryInfC{$ \valtobool{(\denot{e}\,V)} = \BoolT \impl (L,V,s)~({\GSim{p}{r}}\cup r)~(L', V', s') $}
      \noLine
      \UnaryInfC{$ \valtobool{(\denot{e}\,V)} = \BoolF \impl (L,V,t)~({\GSim{p}{r}}\cup r)~(L', V', t') $}
      \UnaryInfC{$(L,V,\ilIf{e}{s}{t})\!\GSim{p}{r}\!\!(L', V'\!\!, \ilIf{e'}{s'\!}{t'})$}
    \end{topprooftree}
  \end{center}
  \caption{Admissible Rules}
  \label{fig:admissible_rules}
\end{figure}

\begin{lemma}
\label{lem:sim_let_op}
\label{lem:sim_let_call}
\label{lem:sim_cond}
 The rules in \myref{fig:admissible_rules} are admissible.
\end{lemma}
\begin{proof}
We only show \nrule{Sim-Let-Op}. After rewriting with \myref{lem:gf_unfold},
we have to show that $(L,V,\ilLet{x}{e}{s})$ and $(L', V', \ilLet{x'}{e'}{s'})$
are related by $\SimF\,(r\cup{\GSim{p}{r}})\,p$.
Case analysis on $\opeval{e}{V}$.
\begin{itemize}
\item Case $\opeval{e}{V}=v$. We unfold $\SimF$ and show the case \nrule{Bisim-Silent}.
  The two required successor states exist:
  \begin{enumerate}
  \item $(L,V,\ilLet{x}{e}{s})\fevalsg{}{}^{+}(L,V\update{x}{v},s)$
  \item $(L', V', \ilLet{x'}{e'}{s'})\fevalsg{}{}^{+}(L', V'\update{x'}{v}, s')$
  \end{enumerate}
  $(L,V\update{x}{v},s)~({\GSim{p}{r}}\cup r)~(L', V'\update{x'}{v}, s')$ holds
  by assumption, which finishes the case.
\item Case $\opeval{e}{V}=\bot$.
  We unfold $\SimF$ and show the case \nrule{Bisim-Term}.
  Both states are terminal, and the way we defined the result function
  ensures that $(L,V,\ilLet{x}{e}{s})\trmg{\bot}{}$ and $(L',V',\ilLet{x'}{e'}{s'})\trmg{\bot}{}$.
\end{itemize}
\end{proof}
\newcommand{\refsimletop}{\hyperref[lem:sim_let_op]{\nrule{Sim-Let-Op}}}
\newcommand{\refsimletcall}{\hyperref[lem:sim_let_call]{\nrule{Sim-Let-Call}}}
\newcommand{\refsimcond}{\hyperref[lem:sim_cond]{\nrule{Sim-Cond}}}

We prove in general that conditionals can be eliminated if the value of the condition is statically known.

\begin{lemma}
\label{lem:if_elimination}
If
\begin{itemize}
\item $\valtobool{(\denot{e}\,\EVE)} = \bot \impl \denot{e}\,V = \denot{e}\,V'$
\item $\forall v, \denot{e}\,V = v \impl \valtobool{v} = \BoolT \impl \valtobool{(\denot{e}\,\EVE)} \not= \BoolF \impl (L, V, s_1) \GSim{p}{r} (L', V', s'_1)$
 \item $\forall v, \denot{e}\,V = v \impl \valtobool{v} = \BoolF \impl \valtobool{(\denot{e}\,\EVE)} \not= \BoolT \impl (L, V, s_2) \GSim{p}{r} (L', V', s'_2)$
\end{itemize}
then
$$
\begin{array}{lll}
  &(L, V, &\ilIf{e}{s_1}{s_2})\\
\GSim{p}{r}&(L', V',\,& \M{if~}\denot{e}\emptyset=\BoolT\M{~then~}s'_1\\
         &&\M{else~if~}\denot{e}\emptyset=\BoolF\M{~then~}s'_2\\
         &&\M{else~} \ilIf{e}{s'_1}{s'_2}).
\end{array}
$$
\end{lemma}
\begin{proof}
Case analysis on $\valtobool(\denot{e}\,\EVE)$.
\begin{itemize}
\item Case $\valtobool(\denot{e}\,\EVE)=\BoolT$.
  By \hyperref[inl:exp_mon]{monotonicity of expression evaluation}, there is $v$ such that $\denot{e}\,V=v$ and $\valtobool{v}=\BoolT$.
  We apply \hyperref[lem:sim_expansion_closed]{\nrule{Sim-Expansion-Closed}}, reducing only the right side one step and finish with
  the second assumption.
\item Case $\valtobool(\denot{e}\,\EVE)=\BoolF$. Analogous to the previous case.
\item Case $\valtobool(\denot{e}\,\EVE)=\bot$.
  Case analysis on $\denot{e} V$.
  If $\denot{e}\,V=\bot$, both sides are stuck by the first assumption.
  We unfold via \myref{lem:gf_unfold} and use the case \nrule{Sim-Term} of $\SimF$ to show simulation.
  If $\denot{e}\,V=v$, then $\denot{e}\,V'=v$ by the first assumption. Case analysis on $\valtobool{v}$.
  If $\valtobool{v}=\BoolT$ ($\valtobool{v}=\BoolF$) we apply \hyperref[lem:sim_expansion_closed]{\nrule{Sim-Expansion-Closed}} to reduce both sides one step
  and finish with the second (third) assumption.
\end{itemize}
\end{proof}

\section{An inductive proof method}
\label{sec:ind_method}
\newcommand{\pr}[5]{\ensuremath{(#1,#2,#3,#4)}}
\newcommand{\Ay}{\ensuremath{\mathit{A}}}
\newcommand{\Pa}{\ensuremath{\mathit{Param}}}
\newcommand{\Ar}{\ensuremath{\mathit{Arg}}}
\newcommand{\Idx}{\ensuremath{\mathit{Idx}}}
\newcommand{\Img}{\ensuremath{\mathit{Img}}}
\newcommand{\Paf}[3]{\ensuremath{\Pa\,#1\,#2\,#3}}
\newcommand{\Arf}[3]{\ensuremath{\Ar\,#1\,#2\,#3}}
\newcommand{\Idxf}[3]{\ensuremath{\Idx\,#1\,#2\,#3}}
\newcommand{\Imgf}[1]{\ensuremath{\Img\,#1}}
\newcommand{\PR}{\ensuremath{\mathcal{P}}}
\newcommand{\App}{\ensuremath{\mathsf{App}}}
\newcommand{\Appf}[4]{\ensuremath{\App^{#1}\,#2\,#3\,#4}}
\newcommand{\simL}[5]{\ensuremath{\Appf{#1}{#2}{#3}{#4} \mathrel{\incl}\, #5}}
\newcommand{\AIE}{\ensuremath{\Lambda}}
\newcommand{\AIEf}{\ensuremath{\Lambda_f}}
\newcommand{\AIEfp}{\ensuremath{\AIEf}}
\newcommand{\labsim}[5]{\ensuremath{#2 \mathrel{#5} #3 :^{#1} #4}}
We develop an inductive proof method for (bi)similarity,
i.e. for statements of the form $(L,V,s) \GSim{p}{r} (L',V',s')$.
Such proofs require relating function contexts $L,L'$,
and we will use the relation $\labsim{\PR}{L}{L'}{\AIE}{r}$ defined below.
For flexibility, the relation is parameterized by a so called
\emph{proof relation} $\PR$, which describes which functions
in $L,L'$ are related, and how their arguments and parameters differ.
\begin{definition}[Proof Relation]
A proof relation is a tuple $\pr{A}{\Pa}{\Ar}{\Idx}{\Img}$ such that
\begin{enumerate}
\item $\Ay : \M{Type}$
\item $\Pa : A \to \slist{\Var} \to \slist{\Var}\to\Prop$
\item $\Ar : A \to \slist{\Val} \to \slist{\Val}\to\Prop$
\item $\Idx : A\to\Fun\to\Fun\to\Prop$
\end{enumerate}
\end{definition}
The proof relation is indexed by a type $A$, which typically represents program analysis information.
A proof relation defines
conditions on formal parameters ($\Pa$), arguments at function calls ($\Ar$), and function names ($\Idx$).
\begin{figure}
\begin{minipage}{.48\textwidth}
\begin{lstlisting}[caption={},basicstyle=\ttfamily]
fun f (x,y) =
 if (x > 9) then 1
 else f (x+1, y)
in f (3,2)
\end{lstlisting}
\end{minipage}
\begin{minipage}{.48\textwidth}
\begin{lstlisting}[caption={},basicstyle=\ttfamily]
fun f (x) =
 if (x > 9) then 1
 else f (x+1)
in f (3)
\end{lstlisting}
\end{minipage}
\caption{An example program before (left) and after (right) dead variable elimination.}
\label{fig:example}
\end{figure}
For example, consider \myref{fig:example}, which contains a program before (left) and after dead variable elimination (right), an optimization we verify in \myref{sec:DVE}.
The proof relation in \myref{def:proofrel_dve} relates the two versions of the function $f$ and expresses that, for example, $y$ is removed from the function parameters of $f$ because the second parameter is not used in the body of $f$.


\subsection{Relating Function Contexts}
We define the relation $\labsim{\PR}{L}{L'}{\AIE}{r}$, which intuitively means that if two functions from $L,L'$ are related by $\Idx$, then their parameters satisfy the relation $\Pa$ and the functions are equivalent if called with arguments in relation $\Ar$.

\begin{definition}
Given a proof relation $\PR$ and
analysis information context $\AIE$, and function context $L, L'$ we say
$\AIE$ and $L, L'$ are in parameter relation with respect to $\PR$, written \mkdef{$\Paf{\AIE}{L}{L'}$}, if whenever
$\Idxf{\AIEfp}{f}{f'}$ and
$L_f=(V, \slist{x}, s)$ and  $L'_{f'}=(V', \slist{x}', s')$ then $\Paf{\AIEfp}{\slist{x}}{\slist{x}'}$.
\end{definition}


\begin{definition}
Given a proof relation $\PR$, and function context $L, L'$,
and analysis information context $\AIE$,
we define a relation \mkdef{$\Appf{\PR}{\AIE}{L}{L'}$} on configurations such that
\begin{align*}
  & \Appf{\PR}{\AIE}{L}{L'}\,(L, V, \ilGoto{f}{\slist{e}})\, (L', V', \ilGoto{f'}{\slist{e'}})\\
  :\iff ~& \exists V\,V' \slist{x}\,\slist{x}'s\,s', L_f=(V, \slist{x}, s) \land L'_{f'}=(V', \slist{x}', s')\\
&\land  \Idxf{\AIEfp}{f}{f'} \land \Arf{\AIEfp}{(\denot{\slist{e}}\,V)}{(\denot{\slist{e'}}\,V')}\\
&\land|\slist{x}|=|\slist{e}|\land|\slist{x}'|=|\slist{e}'|
\end{align*}
The relation $\Appf{\PR}{\AIE}{L}{L'}$ relates application configurations that satisfy the requirements imposed by the proof relation.

\end{definition}

\begin{definition}
Two function contexts $L,L'$ are in $r$-relation with respect to $\AIE$ and $\PR$,
 written \mkdef{\labsim{\PR}{L}{L'}{\AIE}{r}} if
  \begin{enumerate}
  \item $\dom{\AIE}=\dom{L}$
  \item $\Paf{\AIE}{L}{L'}$
  \item $\Idxf{\AIEfp}{f}{f'} \impl (f \in \dom{L} \iff f' \in \dom{L'})$
  \item $\simL{\PR}{\AIE}{L}{L'}{r}$
  \end{enumerate}
\end{definition}

\begin{lemma} \label{lem:labsim_mon}
If $r\incl r'$ and $\labsim{\PR}{L}{L'}{\AIE}{r}$ then $\labsim{\PR}{L}{L'}{\AIE}{r'}$.
\end{lemma}

\subsection{Extending Related Function Contexts}
\label{sec:ext_fun_ctx}
We prove the central lemma of our inductive proof method,
which we call the \emph{extension lemma}.
When descending under function definitions, related $L$ and $L'$ are extended with new closures.
The inductive hypothesis provides that the bodies of these functions are (bi)similar,
but this does not readily mean that the corresponding semantic fixed-points are (bi)similar.
The extension lemma (\myref{lem:sem_extension}) accounts for the semantics of the fixed-point operator.

\newcommand{\BdyF}{\ensuremath{\mathsf{Bdy}}}
\newcommand{\BdyFf}[6]{\ensuremath{\BdyF^{#1}_{#5,#6}\,#2\,#3\,#4}}
\begin{definition}
  \label{def:bdyf}
  Given a proof relation $\PR$, function context $L, L'$ and $K,K'$,
  and analysis information $\AIE$,
  we define a relation \mkdef{$\BdyFf{\PR}{\AIE}{K}{K'}{L}{L'}$} on configurations such that
  \begin{align*}
    &\BdyFf{\PR}{\AIE}{K}{K'}{L}{L'}\,(K\cc L, V\update{\slist{x}}{v}, s)\,(K\cc L', V'\update{\slist{x}'}{v'}, s')\\
    &:\iff~\exists f\,f'\,V\,V', K_f=(V,\slist{x}, s) \land K'_{f'}=(V',\slist{x}', s')\\
    & \quad\quad \land \Idxf{\AIEfp}{f}{f'} \land      \Arf{\AIEfp}{\slist{v}}{\slist{v'}}
  \end{align*}
\end{definition}
The relation $\BdyFf{\PR}{\slist{a}}{F}{F'}{K}{K'}$ relates configurations that are obtained by one reduction from application configurations that satisfy the requirements imposed by the proof relation.
We set up our inductive proofs such that the inductive hypothesis provides that
these configurations are equivalent.

\newcommand{\sep}[4]{\ensuremath{#3 \parallel\!#2\!\parallel_{#1} #4}}
\begin{definition}
  A proof relation $\PR$ separates
  two contexts $K, K'$ under $\AIE;\AIE'$, written \mkdef{$\sep{\PR}{\AIE;\AIE'}{K}{K'}$} if:
  \begin{enumerate}
  \item  $\dom{\AIE}=\dom{K}$
  \item $\Idxf{(\AIE;\AIE')_f}{f}{f'} \impl (f \in \dom K \iff f' \in \dom K')$
  \end{enumerate}
\end{definition}

\begin{lemma}[Extending Parameter Relations]
If we have $\Paf{\AIE}{K}{K'}$ and also $\sep{\PR}{\AIE;\AIE'}{K}{K'}$ and $\Paf{\AIE'}{L}{L'}$ then $\Paf{(\AIE\cc \AIE')}{(K\cc L)}{(K'\cc L')}$.
\end{lemma}
Separation requires that functions in $K$ are only related to functions in $K'$, and vice versa.

\newcommand{\simI}[7]{\ensuremath{\BdyFf{#1}{#2}{#3}{#4}{#5}{#6} \mathrel{\incl}\, #7}}
\begin{lemma}
\label{lem:extension}
Let $\PR$ be a proof relation. If
\begin{enumerate*}
\item $\sep{\PR}{\AIE;\AIE'}{K}{K'}$,
\item $\Paf{\AIE}{K}{K'}$,
\item \simI{\PR}{(\AIE\cc\AIE')}{K}{K}{L}{L'}{(\GSim{p}{r} \cup \mathrel{r})}, and
\item \labsim{\PR}{L}{L'}{\AIE'}{\GSim{p}{r}}
\end{enumerate*}
then
$\labsim{\PR}{K\cc L}{K\cc L'}{\AIE\cc\AIE'}{\GSim{p}{r}}$.
\end{lemma}
\begin{proof}
The proof distinguishes whether the function pair is from $K$ and $K'$ or $L$ and $L'$.
This is possible because $\PR$ separates $K,K'$ under $\AIE$.
In the first case, the result follows from (3) after a lock-step simulation step that reduces function applications on both sides.
In the second case, the result follows from (4) and \hyperref[lem:sim_Y]{\nrule{Sim-Retract}}.
\end{proof}

\newcommand{\fdefsim}[7]{\ensuremath{#2\,|\,#3\vdash#4 \mathrel{#7} #5 :^{#1} #6}}
\newcommand{\fdefsims}[7]{\ensuremath{#2\,|\,#3\vdash\\#4 \mathrel{#7} #5 :^{#1} #6}}
\begin{definition}
\label{def:fctx_ass}
Given a proof relation $\PR$, function context $K, K'$ are in $r$-relation under $L$ and $L'$ with respect to $\PR$ and $\AIE\cc\AIE'$,
 written \mkdef[{\fdefsims{\PR}{L}{L'}{K}{K'}{\AIE\cc\AIE'}{r}}]{$\fdefsim{\PR}{L}{L'}{K}{K'}{\AIE\cc\AIE'}{r}$} if
 \begin{enumerate}
 \item $\sep{\PR}{\AIE;\AIE'}{K}{K'}$
 \item $\Paf{\AIE}{K}{K'}$
 \item $\forall r, \labsim{\PR}{K\cc\! L}{K'\cc\! L'}{\!\AIE\cc\AIE'}{r}\impl \simI{\PR}{K}{K'}{(\AIE\cc\AIE')}{L}{L'}{r}$
 \end{enumerate}
\end{definition}

\begin{lemma}[Fix Compatibility]
\label{lem:fix_compat}
Let $\PR$ be a proof relation. If
$\fdefsim{\PR}{L}{L'}{K}{K'}{\AIE\cc\AIE'}{\GSim{p}{r}}$
and $\labsim{\PR}{L}{L'}{\AIE'}{\GSim{p}{r}}$
then $\simI{\PR}{K}{K'}{(\AIE\cc\AIE')}{L}{L'}{\GSim{p}{r}}$.
\end{lemma}
\begin{proof}
By coinduction via \myref{lem:gf_ind}.
We have to show  \simI{\PR}{F}{F'}{(\AIE\cc\AIE')}{L}{L'}{\SimN_{r'}} from
${r\incl r'}$
and the coinductive hypothesis ${\simI{\PR}{F}{F'}{(\AIE\cc\AIE')}{L}{L'}{r'}}$.
Applying clause (3) of the first premise reduces the proof obligation to $\labsim{\PR}{K\cc L}{K\cc L'}{\AIE\cc\AIE'}{\GSim{p}{r'}}$. We apply \myref{lem:extension}.
The third premise of \myref{lem:extension} follows from the coinductive hypothesis and $r' \incl \SimN_{r'} \cup \mathrel{r'}$, the fourth premise follows from monotonicity (\myref{lem:gf_mon} and \myref{lem:labsim_mon}).
\end{proof}
\myref{lem:fix_compat} shows that equivalence of function bodies is sufficient to
show that the corresponding recursive functions are equivalent.
\begin{lemma}[Extension]
\label{lem:sem_extension}
Let $\PR$ be a proof relation. If we have
$\fdefsim{\PR}{L}{L'}{K}{K'}{\AIE\cc\AIE'}{\GSim{p}{r}}$
and $\labsim{\PR}{L}{L'}{\AIE'}{\GSim{p}{r}}$
then $\labsim{\PR}{K\cc L}{K'\cc L'}{\AIE\cc\AIE'}{\GSim{p}{r}}$.
\end{lemma}
\begin{proof}
We apply \myref{lem:extension}. The only non-trivial premise is to show
$$\simI{\PR}{K}{K'}{(\AIE\cc\AIE')}{L}{L'}{{\GSim{p}{r}} \cup r}$$
We make use of the fact ${\GSim{p}{r}} \incl {\GSim{p}{r}} \cup r$  and finish the proof with \myref{lem:fix_compat}.
\end{proof}

\begin{lemma}[Fun Compatibility]
\label{lem:sim_fun}
Let $\PR$ be a proof relation. If
$\fdefsim{\PR}{L}{L'}{\mkBlocks{F}_V}{\mkBlocks{F'}_{V'}}{\AIE\cc\AIE'}{\GSim{p}{r}}$
and $\labsim{\PR}{L}{L'}{\AIE'}{\GSim{p}{r}}$ and
\begin{align*}
  \forall r,~ &\labsim{\PR}{\mkBlocks{F}_V\cc L}{\mkBlocks{F'}_{V'}\cc L'}{\AIE\cc\AIE'}{\GSim{p}{r}}\impl \\
  &(\mkBlocks{F}_V\cc L,V,t)\GSim{p}{r}(\mkBlocks{F'}_{V'}\cc L',V',t')
\end{align*}
then $(L,V,\ilLetRecP{F}{t})\GSim{p}{r}(L',V',\ilLetRecP{F'}{t'})$.
\end{lemma}
\begin{proof}
We reduce both sides one step.
We apply the last premise and have to show $\labsim{\PR}{\mkBlocks{F}_V\cc L}{\mkBlocks{F'}_{V'}\cc L'}{\AIE\cc\AIE'}{\GSim{p}{r}}$.
\myref{lem:sem_extension} finishes the proof.
\end{proof}
When using \myref{lem:sem_extension} or \myref{lem:sim_fun} it suffices
to show that the bodies of new function definitions are related according to \myref{def:fctx_ass}.
Item (3) already provides that the contexts containing the new functions are related.
See the function definition case of \myref{lem:simeq_refl} to understand how
this enables the inductive method.

\subsection{Using Related Function Contexts to Prove the Application Case}

\begin{definition}
Argument evaluation of $L, V, \slist{e}$ and $L', V', \slist{e}'$ agrees with respect to $p$ and $\PR$ if whenever
$L_f=(V, \slist{x},s)$ and $L_{f'}=(V, \slist{x}',s')$ and $\Paf{a}{\slist{x}}{\slist{x}'}$ then
\begin{enumerate}
  \item if $\denot{\slist{e}}\,V=\slist{v}$ and $|\slist{x}|=|\slist{v}|$ then there exists $\slist{v}'$ such that
    $\denot{\slist{e}'}\,V=\slist{v}'$ and $|\slist{x}'|=|\slist{v}'|$ and $\Arf{a}{\slist{v}}{\slist{v}'}$
  \item if $p=\STBi$ and $\denot{\slist{e}}\,V=\bot$ then $\denot{\slist{e}'}\,V'=\bot$
  \item if $p=\STBi$ and $|\slist{x}|\not=|\slist{e}|$ then $|\slist{x}'|\not=|\slist{e}'|$
\end{enumerate}
\end{definition}

\begin{lemma}
\label{lem:labenv_app}
Let $\PR$ be a proof relation. If
\begin{enumerate}
  \item $\labsim{\PR}{L}{L'}{\AIE}{\GSim{p}{r}}$
  \item $\Idxf{\AIEfp}{f}{f'}$
  \item argument evaluation of $L, V, \slist{e}$ and $L', V', \slist{e}'$ agrees with respect to $p$ and $\PR$
\end{enumerate}
then $(L, V, \ilGoto{f}{\slist{e}})\GSim{p}{r}(L', V', \ilGoto{f'}{\slist{e}'})$.
\end{lemma}
\begin{proof}
  From (2) we have that $\AIEfp$ is defined, and by
  definition of (1) $\dom{\AIE}=\dom{L}$, which means $f \in\dom{L}$,
  hence again by definition of (2) $f' \in\dom{L'}$.
  We assume that $L_f=(V,\slist{x},s)$ and $L_{f'}=(V,\slist{x}',s')$.
  By definition of (1) we know $\Paf{\slist{a}}{L}{L'}$, hence $\Paf{\AIEfp}{\slist{x}}{\slist{x'}}$.
  Case analysis.
  \begin{itemize}
    \item Case $\opeval{\slist{e}}{V}=\slist{v}$.
      \begin{itemize}
      \item If $|\slist{x}|=|\slist{e}|$ we exploit clause (1) of premise (3) and obtain the fact
        $$\Appf{\PR}{\AIE}{L}{L'}\,(L, V, \ilGoto{f}{\slist{e}})\,(L', V', \ilGoto{f'}{\slist{e}'})$$
        We know $\simL{\PR}{\AIE}{L}{L'}{r}$ from premise (1) and are done.
      \item If $|\slist{x}|\not=|\slist{e}|$.
      If $p=\STSi$, we are done using \nrule{Sim-Error}.
      If $p=\STBi$, we exploit clause (3) of assumption (3) and obtain that
      $|\slist{x}'|\not=|\slist{e}'|$. Both sides are stuck (\nrule{Sim-Term}).
    \end{itemize}
    \item Case $\opeval{\slist{e}}{V}=\bot$.
      If $p=\STSi$, we are done by \nrule{Sim-Error}.
      If $p=\STBi$, we exploit clause (2) of assumption (3) and obtain that
      $\opeval{\slist{e}'}{V'}=\bot$. Both sides are stuck.
  \end{itemize}
\end{proof}

\section{Bisimilarity and Similarity are Contextual}
\label{sec:ctx}
In this section we use the inductive method to show that bisimilarity and similarity
are contextual, that is, sound for contextual equivalence and contextual approximation.

\newcommand{\sseq}{\ensuremath{\M{eq}}}
\newcommand{\SimEq}[2]{\ensuremath{\mathrel{\simeq^{#1}_{#2}}}}

\begin{definition}
We define the proof relation $\PR_{\sseq}$ where
\begin{align*}
A&:=\M{list}\,\Var\\
\Paf{\slist{x}}{\slist{y}}{\slist{y}'}&:=\slist{x}=\slist{y}\land\slist{y}=\slist{y}'\\
\Arf{\slist{x}}{\slist{v}}{\slist{v}'}&:=V=V'\land\slist{v}=\slist{v}'\land|\slist{x}|=|\slist{v}|\\
\Idxf{\_}{f}{f'}&:=f=f'
\end{align*}
\end{definition}
\newcommand{\paprj}[1]{\ensuremath{\M{pa}\,#1}}
Let \mkdef{$\paprj{L}$} denote the projection of each closure to its parameters, and
\mkdef{$\paprj{F}$} the projection of each function definition to its parameters.
\begin{definition}[Program Equivalence] For two terms $s, s'$ we define equivalence \mkdef{$s\SimEq{p}{r} s'$} as
\begin{align*}
\forall L L' V, \labsim{\PR_{\sseq}}{L}{L'}{(\paprj{L'})}{\GSim{p}{r}}\impl (L,V,s) \GSim{p}{r} (L',V,s')
\end{align*}
\end{definition}

\begin{lemma}[Reflexivity]
\label{lem:simeq_refl}
$s\SimEq{p}{r}s$.
\end{lemma}
\begin{proof}
By induction on $s$.
\begin{itemize}
\item The case for let follows from the inductive hypothesis, and lemmas \refsimletop{} and \refsimletcall{}.
\item The conditional case follows by \refsimcond{} and the inductive hypotheses.
\item In the case of application $\ilGoto{f}{\slist{e}}$, we do a case analysis on whether $L'_f$ exists.
  If it exists we are done by $\labsim{\PR_{\sseq}}{L}{L'}{\paprj{L'}}{\GSim{p}{r}}$ with \myref{lem:labenv_app},
  and the fact that argument evaluation obviously agrees.
  Otherwise, we know from $\labsim{\PR_{\sseq}}{L}{L'}{\paprj{L'}}{\GSim{p}{r}}$ that $\dom{L}=\dom{L'}$,
  so $L_f$ does not exist either, and both sides are stuck. We finish with \nrule{Sim-Term}.

\item The case for operation is by \nrule{Sim-Term}.
\item In the function definition case we apply \myref{lem:sim_fun}.
  The second premise of \myref{lem:sim_fun} is an assumption, and the third is the inductive hypothesis.
For the first premise $\fdefsim{\PR}{L}{L'}{\mkBlocks{F}_V}{\mkBlocks{F'}_{V'}}{\AIE\cc\AIE'}{\GSim{p}{r}}$ (\myref{def:fctx_ass}) we need to show separation and parameter relation,
which hold because both sides define the same functions.
It remains to show (3) of \myref{def:fctx_ass}, i.e. from
$$(*)~\labsim{\PR}{\mkBlocks{F}_V\cc L}{\mkBlocks{F}_V\cc L'}{(\paprj{F};\paprj{L'})}{\GSim{p}{r}}$$
that $\simI{\PR}{\mkBlocks{F}_V}{\mkBlocks{F}_V}{(\paprj{F};\paprj{L'})}{L}{L'}{\GSim{p}{r}}$ holds.
Unfolding \M{Bdy} (\myref{def:bdyf}), we get $\Idxf{(\paprj{F};\paprj{L'})}{f}{f'}$ and $F_f=(\slist{x}, s)$ and $F_{f'}=(\slist{x}', s')$ and $\Arf{\slist{x}}{\slist{v}}{\slist{v'}}$.
After unfolding $\Idx$ to obtain $f=f'$, and after unfolding $\Ar$ to obtain $\slist{v}=\slist{v}'$ we have to show
that
 $$ (\mkBlocks{F}_V\cc L, V\update{\slist{x}}{v}, s) \Bisim_r (\mkBlocks{F}_{V}\cc L', V\update{\slist{x}}{v}, s)$$
This follows from the inductive hypothesis, but only with~$(*)$.
Note that if \myref{lem:sim_fun} had not provided $(*)$ through \myref{def:fctx_ass},
the proof would not work.
\end{itemize}
\end{proof}

\begin{lemma}[Reflexivity]
$\labsim{\PR}{L}{L}{\paprj{L}}{\GSim{p}{r}}$.
\end{lemma}

\begin{lemma}[Transitivity]
${s\SimEq{p}{\bot}\!s'\impl s'\SimEq{p}{\bot}\!s'' \impl s\SimEq{p}{\bot}\!s''}$.
\end{lemma}

\begin{theorem}
Let $C[]$ be an IL context, i.e. a term with a hole.
$(\forall r, s\SimEq{p}{r}s') \impl C[s]\SimEq{p}{r}C[s']$.
\end{theorem}
\begin{proof}
Induction on $C$.
\end{proof}

\begin{lemma}
If we have $\forall r\, i, s_i\SimEq{p}{r}s_i'$ and $\forall r, t\SimEq{p}{r}t'$ then $\forall r, \ilLetRecM{f}{\slist{x}}{s}{t} \SimEq{p}{r} \ilLetRecM{f}{\slist{x}}{s'}{t'}$.
\end{lemma}






\section{Unreachable Code Elimination}
\newcommand{\uc}[3]{\ensuremath{#1\vdash\textbf{\textup{reach}\,}{}#3}}
\newcommand{\ucc}[4]{\ensuremath{#2\vdash\textbf{\textup{reach}}\,{}#4}}
\newcommand{\LC}{\ensuremath{\Lambda}}
\newcommand{\lv}{\ensuremath{X}}

\label{sec:UCE}
We apply our inductive proof method in the correctness proof of the first optimization: unreachable code elimination.

\subsection{Representing Program Analysis Information}
\label{sec:ann}
\newcommand{\Ann}[1]{\ensuremath{\textsf{Ann}\,#1}}
\newcommand{\AnnT}[2]{\ensuremath{\textsf{Ann}\,#1\,#2}}
\newcommand{\ann}[1]{#1}
\newcommand{\anni}[2]{#1\cdot#2}
\newcommand{\annii}[3]{#1\cdot#2,#3}
\newcommand{\annF}[3]{#1\cdot#2,#3}
\newcommand{\ia}[1]{\set{#1}~}

A program analysis associates information with every subterm of a program.
From now on, we use annotated terms instead of terms, i.e., $\ia{a} s$, which associates analysis information $a$ with term $s$.
Note that the subterms in $s$ are again annotated.
Given analysis information of type $A$, we denote this new inductive type by $\AnnT\Exp{A}$.
The formal development contains verified analyses for DVE and UCE.

\newcommand{\beq}[1]{\ensuremath{=}}

\begin{figure}
\begin{center}
  \begin{topprooftree}{Reach-Let}
    \AxiomC{$b \beq{p} b'$}
    \AxiomC{$\ucc{p}{\LC}{a}{\ia{b'}s}$}
    \BinaryInfC{$\ucc{p}{\LC}{\anni{b}{a}}{\ia{b}\ilLet{x}{\extexp}{\ia{b'}s}}$}
  \end{topprooftree}
  \begin{topprooftree}{Live-Exp}
    \AxiomC{$$}
    \UnaryInfC{$\ucc{p}{\LC}{b}{\ia{b} e}$}
  \end{topprooftree}
  \begin{topprooftree}{Reach-App}
    \AxiomC{$b \impl \LC_f$}
    \UnaryInfC{$\ucc{p}{\LC}{b}{\ia{b}\ilGoto{f}{\slist{e}}}$}
  \end{topprooftree}
\end{center}
\begin{center}
  \begin{topprooftree}{Reach-Cond}
    \AxiomC{$\valtobool(\denot{e}\emptyset) \not=\BoolF \impl b \beq{p} b_1$}
    \noLine
    \UnaryInfC{$\valtobool(\denot{e}\emptyset) \not=\BoolT \impl b \beq{p} b_2$}
    \AxiomC{$\ucc{p}{\LC}{a_1}{\ia{b_1}s_1}$}
    \noLine
    \UnaryInfC{$\ucc{p}{\LC}{a_2}{\ia{b_2}s_2}$}
    \BinaryInfC{$\ucc{p}{\LC}{\annii{b}{a_1}{a_2}}{\ia{b}\ilIf{e}{\ia{b_1}s_1}{\ia{b_2}s_2}}$}
  \end{topprooftree}
\end{center}
\begin{center}
  \begin{topprooftree}{Reach-Fun}
    \Axiom$\fCenter\forall g,~\ucc{p}{{\slist{f:b}};\LC}{a_f}{\ia{b_g}s_g}$
    \noLine
    \UnaryInf$\fCenter\ucc{p}{{\slist{f:b}};\LC}{a'}{\ia{c}t}$
    \Axiom$\fCenter a \beq{p} c$
    \BinaryInfC{$\ucc{p}{\LC}{\annF{b}{\slist{a}}{a'}}{\ia{a}\ilLetRecM{f}{\slist{x}}{\ia{b}s}{\ia{c}t}}$}

  \end{topprooftree}
\end{center}
\caption{Definition of the Reachability Predicate $\ucc{p}{\LC}{a}{s}$. The context $\LC:\textit{context}~\Bool$ contains reachability information for functions and $s:\AnnT{\Exp}{\Bool}$ is an program annotated with reachability information.}
\label{fig:reach}
\end{figure}

\subsection{Inductive Reachability Judgment}
We want to annotate $t$ with reachability information $b,b':\Bool$ such that
whenever $(L,V,\ia{b}s)\red^\ast(L',V',\ia{b'}t)$ we have $b\impl b'$.
Reachability is a non-trival semantic property, hence undecidable.
We define inductively the judgment \ndef{reach} in \myref{fig:reach}.
Reach considers the value of the condition expression $e$ in the empty environment: $\denot{e}\,\emptyset$.
If $\denot{e}\,\emptyset=\bot$, both cases are assumed to be reachable, which may over-approximate.
\nrule{Reach-Let} propagates reachability information through let-bindings:
If the let is reachable, then so is its successor.
\nrule{Reach-App} ensures that whenever a function application is reachable,
then the function is also reachable.
\nrule{Reach-Cond} evaluates $\denot{e}\emptyset$, i.e. it evaluates the condition under the empty variable environment.
If $\valtobool(\denot{e}\emptyset)=\BoolF$ then reachability is not propagated into the consequence~$s_1$.
Propagation into the alternative $s_2$ is treated similarily.
\nrule{Reach-Fun} propagates reachability into $t$.
The context $\LC$ is extended with the reachability information of the function bodies $\slist{f:b}$.
The topmost premise ensures that the reachability information for all function bodies is sound.

\subsection{Transformation and Correctness}
\newcommand{\uceN}{\M{uce}}
\newcommand{\uce}[3]{\ensuremath{\uceN\,#2}}
\newcommand{\countTrueN}{\M{countTrue}}

\newcommand{\filterN}{\M{filter}}
\newcommand{\filter}[2]{\ensuremath{\filterN\,#1\,#2}}
\newcommand{\filterbyN}{\M{filterby}}
\newcommand{\filterby}[3]{\ensuremath{\filterbyN\,#1\,#2\,#3}}

\newcommand{\zipN}{\M{zip}}
\newcommand{\zip}[2]{\ensuremath{\zipN\,#1\,#2}}
In \myref{fig:uce}, we define a function $\uceN$ that removes all code not marked reachable.
If all functions from a mutually recursive function definition are removed, the $\texttt{fun}$-statement is removed, too.
Conditionals are removed if the value of the condition can be statically evaluated.

%
\newcommand{\uceFN}{\M{uceF}}
\newcommand{\uceF}[3]{\ensuremath{\uceFN\,#2}}
%
\begin{figure}
\begin{align*}
  &\filterbyN\,\,:~\forall X Y, (X\to\Bool)\to \M{list}\,X\to\M{list}\,Y\to\M{list}\,Y\\
  &\filterbyN\,p\,(x::x')\,(y::y')=\\
  &\quad\M{if~}p x\M{~then~}y::\filterby{p}{x'}{y'}\M{~else~}\filterby{p}{x'}{y'}\\
  &\filterbyN\,p\,\_,\_=\M{nil}\\[1mm]
  &\filterN\,\,:\forall X, (X\to\Bool)\to \M{list}\,X\to\M{list}\,X\\
  &\filterN\,p\,x=\filterby{p}{x}{x}
\end{align*}
\caption{Definition of \M{filter}}
\label{fig:helpers}
\end{figure}

\begin{figure}
\begin{align*}
  &\uceN:\AnnT{\Exp}{\Bool}\to\Exp\\
  &\uce{B}{(\ilLet{x}{\extexp}{s})}{(\anni{b}{a})}~=~\ilLet{x}{\extexp}{(\uce{L}{s}{a})}\\
  &\uce{B}{(\ilIf{e}{s_1}{s_2})}{(\annii{b}{a_1}{a_2})}~=~\\
         &\quad\M{if~}\denot{e}\emptyset=\BoolT\M{~then~}\uce{L}{s_1}{a_1}\\
         &\quad\M{else~if~}\denot{e}\emptyset=\BoolF\M{~then~}\uce{L}{s_2}{a_2}\\
         &\quad\M{else~} \ilIf{e}{(\uce{L}{s_1}{a_1})}{(\uce{L}{s_2}{a_2})}\\
  &\uce{\slist{b}}{(\ilGoto{f}{\slist{e}})}{b}~=~
  \ilGoto{f}{\slist{e}}\\
  &\uce{B}{e}{\_}~=~e\\
  &\uce{B}{(\ilLetRecP{F}{t})}{(\annF{\_}{\slist{a}}{b})}{}~=~\\
   &\quad \M{let~} F' = \uceF{(\cg{\slist{b}};L)}{F}{\slist{a}}~\M{in} \\
   &\quad \M{if~} |F'| = 0 \M{~then~} \uce{([\slist{a}];L)}{t}{b} \M{~else~}
           \ilLetRecP{F'}{(\uce{([\slist{a}];L)}{t}{b})}\\
  &\uceF{L}{F}{\slist{a}}~=~\\
&\quad \M{let~} K = \filter{(\lambda(\slist{x},\ia{b}s).\,b)}{F}\M{~in}\\
  &\quad \M{map}\,(\lambda(\slist{x},s).(\slist{x},\uce{L}{s}{a}))\,K
\end{align*}
\caption{Definition of Unreachable Code Elimination}
\label{fig:uce}
\end{figure}

\begin{definition}
We define the proof relation $\PR_{\uceN}$ where
\begin{align*}
A&:=\Bool\\
\Paf{\_}{\slist{x}}{\slist{x}'}&:=\slist{x}=\slist{x}'\\
\Arf{a}{\slist{v}}{\slist{v}'}&:=\slist{v}=\slist{v}'\\
\Idxf{a}{f}{f'}&:=a = \BoolT \land f = f'
\end{align*}
\end{definition}

\begin{lemma}
\label{lem:uce_sep}
If $\dom{\AIE}=\dom{F}$ then
$\sep{\PR_\uceN}{\AIE;\AIE'}{F}{\uceF{([\slist{a}];L)}{F}{\slist{a}}}$.
\end{lemma}

\begin{lemma}
\label{lem:uce_pa}
$F=\ilFDefs{f}{\slist{x}}{\ia{b}s}\impl\Paf{\slist{b}}{F}{(\uceF{(\cg{\slist{b}};L)}{F}{\slist{a}})}$.
\end{lemma}

\begin{lemma}
\label{lem:uce_fun_cases}
Suppose that it holds $F'=\uceF{([\slist{a}];L)}{F}{\slist{a}}$ and
we have that $(\mkBlocks{F}_V\cc L, V, s)\Bisim_r(\mkBlocks{F'}_{V'}\cc L', V', \uce{([\slist{a}];L)}{t}{b})$
then
\[\arraycolsep=0pt
  \begin{array}{lll}
    (L, V, \ilLetRecP{F}{s})
    \Bisim_r~(L', V',~&\M{if~} |F'| = 0 \M{~then~}\uce{([\slist{a}];L)}{t}{b}\\
   &\M{~else~} \ilLetRecP{F'}{\uce{(\slist{a};L)}{t}{b}}).
  \end{array}\]
\end{lemma}
\begin{proof}
If $|F'|=0$, then $F'=\M{nil}$. After reducing only the left side one step via \hyperref[lem:sim_expansion_closed]{\nrule{Sim-Expansion-Closed}}, the assumption solves the goal.
Otherwise we reduce both sides one step (\myref{lem:gf_unfold}), and the assumption solves the goal.
\end{proof}

\begin{theorem}
\label{thm:uce_correct}
Let $\ucc{\textsf{s}}{\AIE}{a}{\ia{\BoolT}s}$ and $\labsim{\PR_{\uceN}}{L}{L'}{\AIE}{\Bisim_r}$.
Then: $(L, V, s)\Bisim_r(L', V, \uce{\slist{b}}{s}{a})$.
\end{theorem}
\begin{proof}
Induction on $s$ and in each case inversion of \textbf{reach}.
\begin{itemize}
\item The case for let follows from \refsimletcall{} and the inductive hypothesis.
\item The case for the conditional follows by \myref{lem:if_elimination} and the inductive hypotheses.
\item The application case follows from $\labsim{\PR_{\uceN}}{L}{L'}{\AIE}{\Bisim_r}$ with \myref{lem:labenv_app}, after discharging premises:
  Note that $\AIEfp=\BoolT$ by inversion on \ndef{reach}, so
  $\Idxf{\AIEfp}{f}{f}$ holds by definition.
  Argument evaluation agrees, since parameters,
  and arguments and environments are identical.

\item The case for operation is trivial, since operation and environments
  are identical.
\item In the function definition case, let $F'=\uceF{(\slist{a};L)}{F}{\slist{a}}$.
\myref{lem:uce_fun_cases} lets us deal with both cases uniformly and requires
$$(\mkBlocks{F}_V\cc L, V, t)\Bisim_r(\mkBlocks{F'}_E\cc L', E, \uce{(\slist{a};L)}{t}{b})$$
After applying the inductive hypothesis, we must show $\labsim{\PR}{\mkBlocks{F}_E\cc L}{\mkBlocks{F'}_E\cc L'}{\AIE'\cc\AIE}{\SimN_r}$.
We apply \myref{lem:sem_extension} and discharge its premises by using \myref{lem:uce_sep} and \myref{lem:uce_pa}.
The remaining premise requires us to show from
\begin{align*}
\labsim{\PR}{\mkBlocks{F}_E\cc L}{\mkBlocks{F'}_E\cc L'}{\AIE'\cc\AIE}{\BisimN_r}&&(*)
\end{align*}
that $\simI{\PR}{\mkBlocks{F}_E}{\mkBlocks{F'}_E}{(\AIE'\cc\AIE)}{L}{L'}{\BisimN_r}$.
Unfolding \M{Bdy}, we obtain $\Idxf{(\AIE'\cc\AIE)}{f}{f'}$ and  $F_f=(\slist{x}, s)$ and $F'_{f'}=(\slist{x}', s')$ and $\Arf{\AIEf'}{\slist{v}}{\slist{v'}}$.
Unfolding those, we have to show
$$ (\mkBlocks{F}_E\cc L, E\update{\slist{x}}{v}, s) \Bisim_r (\mkBlocks{F'}_E\cc L', E\update{\slist{x}}{v}, \uce{}{s}{})$$
The inductive hypothesis solves the goal with $(*)$.
\end{itemize}
\end{proof}

\section{Dead Variable Elimination}
\label{sec:DVE}
Dead variable elimination relies on a true liveness analysis.
A variable is live if it is (potentially) used later on.
A variable is true life, if it is used to compute a value that is live later on.
True liveness requires a fixed-point computation, but is able to
detect parameters of a function that do not contribute to the behavior of the function.

\subsection{Inductive Liveness Judgement}
\newcommand{\tl}[4]{\ensuremath{#1~|~#2\vdash\textbf{\textup{tlive}\,}{}#4}}
\newcommand{\LZ}{\ensuremath{\zeta}}

We specify sound true liveness information with the inductive judgment \ndef{tlive} in \myref{fig:truelive}.

\nrule{TLive-Op} ensures that all variables that are live after the let ($X'$) are also
live before the let, except the variable defined.
The free variables of $e$ only need to be live if $x$ is live after the let.
\nrule{TLive-Call} is similar, but always requires the free variables of $\slist{e}$ to be live,
as we can never remove calls (even if their result is unused).
\nrule{TLive-Exp} requires the free variables of $\slist{e}$ to be live.
\nrule{TLive-Cond} tests whether the condition is a constant expression by evaluating it
in the empty environment.
Only if this is unsuccessful, we require its free variables to be live.
If $\valtobool(\denot{e}\emptyset)\not=\BoolF$, the consequence might be reachable.
In this case the rule requires that all variables live in the consequence ($X_1$) are live before
the conditional ($\lv_1\incl\lv$), and that the judgment holds recursively.
Otherwise, no requirements are imposed.
\nrule{TLive-App} requires whenever a parameter $x_i$ is in the live set of the function body $X_f$,
then the free variables of the corresponding argument expression $e_i$ are live at the application.
\nrule{TLive-Fun} requires the liveness judgment to recursively hold for the continuation and
the function bodies under extended contexts.
The context $\LZ$ is extended with the parameters $\slist{f:\slist{x}}$ of the newly defined functions,
and the context $\LC$ is extended with the liveness information $\slist{f:Y}$ for the function bodies.
The rule requires all variables live after the function definition to be live before the function definition,
and that all variables live in the function bodies (except parameters) are also live before the function definition.

\begin{figure}
  \begin{center}
    \begin{topprooftree}{TLive-Op}
      \Axiom$\lv' \setminus \set{x} \incl \lv\fCenter$
      \noLine
      \UnaryInf$x \in \lv' \impl \fv(e)\incl \lv\fCenter$
      \AxiomC{$\tl{\LZ}{\LC}{a'}{\ia{\lv'}s}$}
      \BinaryInfC{$\tl{\LZ}{\LC}{\anni{\lv}{a}}{\ia{\lv}\ilLet{x}{e}{\ia{\lv'}s}}$}
    \end{topprooftree}
    \begin{topprooftree}{TLive-Exp}
     \AxiomC{$\fv(e)\incl\lv$}
     \UnaryInfC{$\tl{\LZ}{\LC}{\lv}{\ia{X}e}$}
   \end{topprooftree}
  \end{center}
  \begin{center}
    \begin{topprooftree}{TLive-Call}
      \Axiom$\lv' \setminus \set{x} \incl \lv\fCenter$
      \noLine
      \UnaryInf$\fv(\slist{e})\incl \lv\fCenter$
      \AxiomC{$\tl{\LZ}{\LC}{a'}{\ia{\lv'}s}$}
      \BinaryInfC{$\tl{\LZ}{\LC}{\anni{\lv}{a}}{\ia{\lv}\ilEvent{x}{\alpha}{\slist{e}}{\ia{\lv'}s}}$}
    \end{topprooftree}
   \begin{topprooftree}{TLive-App}
     \AxiomC{$|\slist{x}|=|\slist{e}|$}
     \noLine
     \UnaryInfC{$\forall i, x_i\in\LC_f\impl\fv(e_i)\incl\lv$}
     \UnaryInfC{$\tl{\LZ}{\LC}{\lv}{\ia{\lv}\ilGoto{f}{\slist{e}}}$}
   \end{topprooftree}
  \end{center}
 \begin{center}
 \begin{topprooftree}{TLive-Cond}
   \AxiomC{$\denot{e}\emptyset=\bot \impl \fv(e)\incl\lv$}
   \noLine
   \UnaryInfC{$\valtobool(\denot{e}\emptyset)\not=\BoolF \impl \lv_1\incl\lv \land \tl{\LZ}{\LC}{a_1}{\ia{\lv_1}s_1}$}
   \noLine
   \UnaryInfC{$\valtobool(\denot{e}\emptyset)\not=\BoolT \impl \lv_2\incl\lv \land \tl{\LZ}{\LC}{a_2}{\ia{\lv_2}s_2}$}
   \UnaryInfC{$\tl{\LZ}{\LC}{\annii{\lv}{a_1}{a_2}}{\ia{X}\ilIf{e}{\ia{\lv_1}s_1}{\ia{\lv_2}s_2}}$}
   \end{topprooftree}
 \end{center}
 \begin{center}
   \begin{topprooftree}{TLive-Fun}
     \Axiom$\fCenter\tl{\slist{f:\slist{x}}\cc\LZ}{\slist{f:Y}\cc\LC}{a'}{\ia{\lv'}t}$
     \noLine
     \UnaryInf$\forall g,\,\fCenter\tl{\slist{f:\slist{x}}\cc\LZ}{\slist{f:Y}\cc\LC}{a_g}{\ia{Y_g}s_g}$
     \Axiom$\lv' \incl \lv\fCenter$
     \noLine
     \UnaryInf$\!\!\!\!\!\!\!\!\!\forall g,~ Y_g \setminus \slist{x}_g \incl \lv\fCenter$
     \BinaryInfC{$\tl{\LZ}{\LC}{\annF{\lv}{\slist{a}}{a'}}{\ia{\lv}\ilLetRecM{f}{\slist{x}}{\ia{Y}s}{\ia{\lv'}t}}$}
   \end{topprooftree}
 \end{center}
 \caption{Definition of the True Liveness Predicate $\tl{\LZ}{\LC}{a}{s}$. The context $\LZ:\textit{context}~(\M{list}\,\Var)$ contains parameters of functions, $\LC:\textit{context}~(\M{set}\,\Var)$ contains the variables live in the function body, and $s:\AnnT{\Exp}{(\M{set}\,\Var)}$ is a program annotated with liveness information.}
 \label{fig:truelive}
\end{figure}

\subsection{Transformation and Correctness}

\newcommand{\dveN}{\M{dve}}
\newcommand{\dve}[4]{\ensuremath{\dveN\,#1\,#2\,#3}}
We realize dead variable elimination (DVE) with the recursive function $\dveN$ defined in \myref{fig:dve}.
The recursive procedure descends through the program, removes unused let-bindings and filters
parameter and argument lists according to liveness information.
As in UCE, conditionals are removed if the condition can be statically evaluated.

\begin{figure}
\begin{align*}
  &\dveN:\M{list}\,\Var\to\M{set}\,\Var\to\AnnT{\Exp}{(\M{set}\,\Var)}\to\Exp\\
  &\dve{\LZ}{\LC}{(\ilLet{x}{e}{\ia{X}s})}{(\anni{\lv}{a})}~=~\\
    &\quad\M{let~}s' = \dve{\LZ}{\LC}{(\ia{X}s)}{a}\M{~in}\\
    &\quad\M{if~}x\in\lv\M{~then~}\ilLet{x}{e}{s'}\M{~else~}s'\\
  &\dve{\LZ}{\LC}{(\ilEvent{x}{\alpha}{\slist{e}}{s})}{(\anni{\lv}{a})}~=~\\
    &\quad\ilEvent{x}{\alpha}{\slist{e}}{(\dve{\LZ}{\LC}{s}{a})}\\
  &\dve{\LZ}{\LC}{(\ilIf{e}{s_1}{s_2})}{(\annii{\lv}{a_1}{a_2})}~=~\\
         &\quad\M{if~}\denot{e}\emptyset=\BoolT\M{~then~}\dve{\LZ}{\LC}{s_1}{a_1}\\
         &\quad\M{else~if~}\denot{e}\emptyset=\BoolF\M{~then~}\dve{\LZ}{\LC}{s_2}{a_2}\\
         &\quad\M{else~} \ilIf{e}{(\dve{\LZ}{\LC}{s_1}{a_1})}{(\dve{\LZ}{\LC}{s_2}{a_2})}\\
  &\dve{(\LZ;f:\slist{x};\LZ')}{(\LC;f:\lv;\LC')}{(\ia{\_}\ilGoto{f}{\slist{e}})}{\lv}~=~\\
  &\quad\ilGoto{f}{(\filterby{(\lambda x.x\in\lv)}{\slist{x}}{\slist{e}})}\\
  &\dve{\LZ}{\LC}{e}{\_}~=~e\\
  &\dve{\LZ}{\LC}{(\ilLetRecM{f}{\slist{x}}{\ia{X}s}{t})}{(\annF{\_}{\slist{a}}{b})}{}~=~\\
   &\quad \M{let~} \forall i, F'_i = (\filter{(\lambda x.x\in X_i)}{\slist{x}_i},\\
   &\quad\quad\dve{(\slist{\slist{x}}\cc\LZ)}{(\slist{X};\LC)}{s_i}{a_i})\M{~in}\\
   &\quad \ilLetRecP{F'}{(\dve{(\slist{\slist{x}},\LZ)}{(\slist{X};\LC)}{t}{b})}
\end{align*}
\caption{Definition of Dead Variable Elimination}
\label{fig:dve}
\end{figure}


\begin{definition}
\label{def:proofrel_dve}
We define the proof relation $\PR_{\dveN}$ where
\begin{align*}
A:=&~\M{list}\,\Var\times\M{set}\,\Var\\
\Paf{(\slist{x},\lv)}{\slist{y}}{\slist{y}'}:=&~\slist{x}=\slist{y}\land\slist{y'}=\filter{(\lambda x. x\in\lv)}{\slist{y}}\\
\Arf{(\slist{x},\lv)}{\slist{v}}{\slist{v}'}:=&~\slist{v'}=\filterby{(\lambda (x). x\in\lv)}{\slist{x}}{\slist{v}}\\
&~\land|\slist{x}|=|\slist{v}|\\
\Idxf{\_}{f}{f'}:=&~f = f'
\end{align*}
\end{definition}

\begin{lemma}
\label{lem:dve_sep}
If $\dom{F}=\dom{\AIE}=\dom{F'}$ then we have $\sep{\PR_\dveN}{\AIE;\AIE'}{F}{F'}$.
\end{lemma}

\begin{lemma}
\label{lem:dve_pa}
Let $F=\ilFDefs{f}{\slist{x}}{\ia{X}s}$ and $F'$ s.t. for all $i$,
$$F'_i =  (\filter{(\lambda x.x\in X_i)}{\slist{x}_i},
\dve{(\slist{\slist{x}}\cc\LZ)}{(\slist{X};\LC)}{s_i}{a_i})$$
 and $|F|=|F'|$.
Then $\Paf{(\slist{f:(\slist{x},X)})}{F}{F'}$.
\end{lemma}
We are now ready to show the correctness theorem.
We write \mkdef{$V =_X V'$} if $V$ and $V'$ agree on the values of the variables in the set $X$.
\begin{theorem}
\label{thm:dve_correct}
Let $\tl{\LZ}{\LC}{a}{\ia{\lv}s}$ and $\labsim{\PR_{\dveN}}{L}{L'}{\zip{\LZ}{\LC}}{\Sim_r}$ and $V =_X V'$.
Then: $$(L, V, s)\Sim_r(L', V', \dve{\LZ}{\LC}{\ia{\lv}s}{a}).$$
\end{theorem}
\begin{proof}
Induction on $s$ and in each case inversion of \textbf{tlive}.
\begin{itemize}
\item The case for let-call follows from \refsimletcall{} and the inductive hypothesis.
\item In the case for let op, we do a case analysis in $x\in\lv$.
\begin{itemize}
\item If $x\in\lv$, the case follows from  \refsimletop{} with the fact that $\opeval{e}{V}=\opeval{e}{V'}$ because $V$ and $V'$ agree on $\fv(e)$.
\item If $x\not\in\lv$, case analysis on $\opeval{e}{V}$.
\begin{itemize}
\item If $\opeval{e}{V}=\bot$, the left side is stuck (\nrule{Sim-Error}).
\item If $\opeval{e}{V}=v$, we use \hyperref[lem:sim_expansion_closed]{\nrule{Sim-Expansion-Closed}} to reduce the left side one step and are done by the inductive hypothesis with the observation that
$V\update{x}{v}$ still agrees with $V'$ on the $X$ because $x\not\in\lv$.
\end{itemize}
\end{itemize}

\item The case for the conditional follows by \myref{lem:if_elimination} and the inductive hypotheses.

\item The case for application follows from $\labsim{\PR_{\dveN}}{L}{L'}{\zip{\LZ}{\LC}}{\Bisim_r}$ with \myref{lem:labenv_app}, after discharging premises. The relation
  $\Idxf{(\slist{x},\lv)}{f}{f}$ holds by definition.
  Argument evaluation agrees because
  \begin{align*}
    &\denot{\filterby{(\lambda x.x\in\lv_f)}{\slist{x}}{\slist{e}}}\,V'\\
    =&~\filterby{(\lambda x.x\in\lv_f)}{\slist{x}}{(\denot{\slist{e}}\,V)}
  \end{align*}
  since we already know that $\denot{\slist{e}}\,V\not=\bot$ and $V$ and $V'$ agree on the live variables.

\item The case for operation is trivial, since operation are identical and environments agree on the live variables.
\item In the function definition case, we have $F'$ such that $|F'|=|F|$ and for all $i$
$$F'_i = (\filter{(\lambda x.x\in X_i)}{\slist{x}_i},
\dve{(\slist{\slist{x}}\cc\LZ)}{(\slist{X};\LC)}{s_i}{a_i})$$
We apply \myref{lem:sim_fun} and have to discharge premises.
The second premise holds by assumption, the third is the inductive hypothesis.
The first two requirements of the first premise are \myref{lem:dve_sep} and \myref{lem:dve_pa}.
It remains to show from
$(*) \labsim{\PR}{\mkBlocks{F}\cc L}{\mkBlocks{F'}\cc L'}{\slist{(\slist{x},\lv};\AIE)}{\BisimN_r}$ that
$$\simI{\PR}{F}{F'}{(\slist{(\slist{x},\lv)}\cc\AIE)}{L}{L'}{\BisimN_r}$$
After unfolding \M{Bdy} (\myref{def:bdyf}), we have the assumptions $\Idxf{(\slist{x},X)}{f}{f'}$ and  $F_f=(\slist{x}, s)$ and $F'_{f'}=(\slist{x}', s')$ and $\Arf{(\slist{x},X)}{\slist{v}}{\slist{v'}}$. And after further unfolding we get $\slist{x}'=\filter{(\lambda x.x\in \lv)}{\slist{x}}$ and $\slist{v}'=\filterby{(\lambda (x). x\in \lv)}{\slist{x}}{\slist{v}}$. We have to show
that
\begin{align*}
 &~(\mkBlocks{F}\cc L, V\update{\slist{x}}{v}, s) \\
\Bisim_r &~(\mkBlocks{F'}\cc L', V'\update{\slist{x}'}{\slist{v'}}, \dve{(\slist{\slist{x}}\cc\LZ)}{(\slist{\lv};\LC)}{s}{a_f})
\end{align*}
Inductive hypothesis provides the latter. Its premises are discharged by $(*)$ and the observation that
the updated environments still agree on the live variables.
\end{itemize}
\end{proof}

\section{Coq Development}
\label{sec:coq}
The formal development accompanying this paper is part of a verified compiler LVC.
LVC use the inductive method presented in this paper, and variations of it,
 for many correctness proofs.
These include DVE, UCE, Copy Propagation, Sparse Conditional Constant Propagation,
and some lowering passes.
LVC also features an imperative variant of IL, which is called IL/I and serves as source language.
The difference between IL and IL/I is that the latter uses imperative variables instead of lexically scoped binders.
The first transformations in the pipeline of the LVC compiler are UCE and DVE on IL/I.
The formal development hence also contains proofs of UCE and DVE for IL/I.
The setup as described works for IL/I, too, and the proof structure
remains the same.

In the formal development we use De-Bruijn indices, not a named representation for function binders.
This fact complicates UCE, as indices change whenever a function is removed.
Fortunately, this does not causes problems with our inductive method.

The part of the LVC development that pertains to this paper is available online:
\begin{center}
\url{www.ps.uni-saarland.de/~sdschn/lvc-ind/}
\end{center}
LVC has more than 36k LoC.
The formalization of the inductive method (\myref{sec:ind_method}) takes \textasciitilde 400 LoC.
The correctness theorems (\myref{thm:uce_correct}, \myref{thm:dve_correct}) are \textasciitilde 50 LoC each.
Also counting lemmas, it takes \textasciitilde 300 LoC to show each of DVE and UCE correct.

\section{Conclusion}
\label{sec:conclusion}
We described an inductive method for proofs of simulation-based program equivalence.
In contrast to the standard approach, which indexes the (bi)simulation with a measure, our approach works without modifying the simulation.
With out method, bisimilarity can be proven with the need for symmetrization.

After the method is setup, the overhead of the correctness proofs of transformations is low, and the proof becomes a straight-forward induction.
The details of the proof are simple enough to be explained in full on paper.
The method separates concerns: The correctness proofs follow the syntactic definition of the transformation, and a separate, general lemma proved by coinduction is used to deal with fixed-point computation in the language.
This allows to focus on the actual verification problem inherent to the transformation instead of requirements induced by the proof method.
We think inductive methods like ours are an essential tool for bisimulation-based compiler verification and useful in general.

We applied the method to two optimizations, unreachable code elimination (UCE) and dead variable elimination (DVE).
The optimizations are not straight-forward to verify because they remove instructions and change function signatures.
The inductive method also improves modularity of the correctness argument:
We argued certain removal steps in a separate lemma (e.g.\ \myref{lem:if_elimination}).
After using such a lemma in a plain coinductive proof (even when using Paco~\cite{DBLP:conf/popl/HurNDV13}), further justification would be required before the cohypothesis could be applied.


\printbibliography

\end{document}